\documentclass[11pt]{article}
\usepackage{amssymb,amsfonts, color}
\usepackage{latexsym,color,graphicx,amsmath,amsthm}
\usepackage{enumitem}

\newcommand{\sgn}{\text{sgn}}

\def\R{{\mathbb R}}
\def\S{{\mathbb S}}

\def\K{{\mathbb K}}

\def\x{{\bf x}}

\def\a{{\bf a}}

\def\rank{{\rm rank}}
\def\deg{{\rm deg}}
\def\rk{{\rm rank}}
\def\sp{{\rm span}}
\def\char{{\rm char}}
\def\poly{{\rm poly}}

\newtheorem{thm}{Theorem}[section]
\newtheorem{cor}[thm]{Corollary}
\newtheorem{lem}[thm]{Lemma}
\newtheorem{clm}[thm]{Claim}

\newtheorem{open}[thm]{Open Problem}
\newtheorem{remark}[thm]{Remark}

\theoremstyle{definition}
\newtheorem{defin}[thm]{Definition}

\begin{document}

\title{Subspace arrangements, graph rigidity and derandomization through submodular optimization\thanks{
The first author was partially supported from NSF grant DMS-1128155.
The second author was partially supported from NSF grant CCF-1412958
}}

\date{}
\author{
Orit E. Raz\thanks{%
Department of Mathematics, 
University of British Columbia, Vancouver, Canada.
{\sl oritraz@math.ubc.ca} }
\and 
Avi Wigderson\thanks{%
School of Mathematics, Institute for Advanced Study,
Princeton NJ 08540, U.S.A.
{\sl avi@ias.edu} }}

\maketitle
\noindent
\begin{center}
{\it Dedicated with admiration to L\'{a}szl\'{o} Lov\'asz, \\on the occasion
of his 70th birthday.}
\end{center}

\begin{abstract}

This paper presents a deterministic, strongly polynomial time algorithm for computing the matrix rank for a class of symbolic matrices (whose entries are polynomials over a field). This class was introduced, in a different language, by  Lov\'asz~\cite{Lov} in his study of flats in matroids, and proved a duality theorem putting this problem in $NP \cap coNP$. As such, our result is another demonstration where ``good characterization'' in the sense of Edmonds leads to an efficient algorithm. In a different paper Lov\'asz~\cite{Lov79} proved that all such symbolic rank problems have efficient probabilistic algorithms, namely are in $BPP$.  As such, our algorithm may be interpreted as a derandomization result, in the long sequence special cases of the PIT (Polynomial Identity Testing) problem. Finally, Lov\'asz and Yemini~\cite{LoYe} showed how the same problem generalizes the {\em graph rigidity} problem in two dimensions. As such, our algorithm may be seen as a generalization of the well-known deterministic algorithm for the latter problem.

There are two somewhat unusual technical features in this paper. The first is  the translation of 
Lov\'asz' flats problem into a symbolic rank one. The second is the use of submodular optimization for derandomization. We hope that the tools developed for both will be useful for related problems, in particular for better understanding of graph rigidity in higher dimensions.
\end{abstract}

\section{Introduction}

In this paper we provide a new {\em deterministic}, strongly polynomial time algorithm which can be viewed in two ways. The first is as solving a derandomization problem, providing a deterministic algorithm to a new special case of the PIT (Polynomial Identity Testing) problem. The second is as computing the dimension of the span a collection of subspaces in high dimensional space. Motivating and connecting the two is the problem of testing {\em graph rigidity}, to which an efficient deterministic algorithm is known only in  the plane, and is open for higher dimensions. Accordingly, we will divide the introduction to explain these three  problems.

\subsection{Polynomial Identity Testing (PIT)}

Let $\K$ be a field. Let $\x = (x_1, \dots x_d)$ be a $d$-tuple of independent variables. The PIT problem is to determine, given a multivariate polynomial $p\in \K[\x]$, if $p\equiv 0$ (as a polynomial). Of course, the description of $p$ as an input to this problem is central to its complexity, and many variants of this problem were considered. The most common formulation is when $p$ is given by an arithmetic formula or circuit\footnote{When the input is a circuit, the degree of $p$ is always assumed to be polynomial in the circuit's size, and in all cases considered in this paper this will be evident.}. 

The original version of this question was posed by Edmonds~\cite{Edm67}. In his formulation, $p$ is the determinant  of a matrix whose entries are linear forms in $\x$ (we will refer such a matrix as a {\em symbolic} matrix). 
 Lov\'asz~\cite{Lov79} proved that this problem is in $BPP$ namely has a fast probabilistic algorithm (for fields $\K$ 
larger than the degree of $p$): indeed, the algorithm simply picks  random elements from $\K$ and evaluates $p$ (note that evaluating $p$ is efficient in all three formulations above, and indeed in all formulations considered). This left open the problem of finding an efficient deterministic algorithm, namely derandomizing Lov\'asz's algorithm for PIT.  

\begin{open}
Is PIT $\in P$?
\end{open}

The importance of this seemingly specific open problem was revealed in an important result of Kabanets and Impagliazzo~\cite{KI04}. They showed that if the answer is positive (as everyone expects), this will imply non-trivial lower bounds on either arithmetic or Boolean circuits, well beyond current techniques.

The progress towards resolving this open problem has been by providing deterministic polynomial time algorithms for a large variety of special cases of it, with the idea of building up techniques. By far, in most of these results the special cases are defined by restricting the input polynomial to lie in some complexity class. In these cases, progress in derandomization followed closely progress on lower bounds for the appropriate class
(as is the case in the Boolean setting as well). There are literally dozens of such papers: many are mentioned and explained in the surveys~\cite{Sax13, ShYe10} and e.g. the recent paper~\cite{ASSS16}.

In parallel, with motivation from algebra, geometry and other areas, a different collection of special cases of PIT was studied, of a structural nature. Here one works with Edmond's formulation, and develops an understanding (and often a polynomial time algorithm) for cases where the symbolic matrix has restricted structure. This includes for example the works~\cite{BrLu08, CIK97, FGT16,  Gee99, IKS10, Mesh85}.

This paper contributes to the second line of research, providing new families of symbolic matrices for which PIT can be solved in deterministic polynomial time. To explain this structure we introduce some notation. We will work in a slightly more general setting, in two ways, as the results generalize to both. First, we will allow our symbolic matrices to have polynomial entries. In such cases, these polynomials will have simple formulas describing them. Second, we will be interested in computing the {\em rank} of the input symbolic matrix, not just whether its determinant vanishes. While seemingly a more general problem, this turns out to be equivalent to PIT (see e.g.  \cite[Appendix A]{GGOW15}\footnote{The proof in \cite{GGOW15} is given for {\em non-commutative} rank, but the exact same proof works verbatim for our usual notion of rank over $\K(\x)$.}).

Let $R$ be a family of polynomial maps $R = \{ r: \K^d  \rightarrow \K^n \}$. In all cases we assume the degree of all polynomials in all maps is at most $n$, and the number of variables $d$ is at most polynomial in $n$, so we will think of $n$ as the input size to the problem. 

A family of maps $R$  prescribes  a family of symbolic matrices, so that each row is an image of the $d$-vector of variables $\x$ under some map in $R$.
More formally, define PIT($R$) to be the set of all symbolic matrices $M$ (with $n$ columns, and $\poly(n)$ rows) in which every row of the matrix is of the form $r(\x)$, for some map $r\in R$. We will be interested in families $R$ for which the ranks of matrices in PIT($R$) can be computed  in polynomial time\footnote{We identify the set of matrices and the computational problem of determining their ranks.}. 

We first demonstrate the convenience of this notation. Call $R$ {\em complete}, if a deterministic polynomial-time algorithm for PIT($R$) implies a deterministic polynomial-time algorithm for PIT. Very simple maps are complete! It follows from Valiant's~\cite{Val79} hardness  of the determinant for the class\footnote{The arithmetic analog of the Boolean class $P$.} {\em VP} that 

\begin{thm}[\cite{Val79}]
The class $R_{\rm affine}$ of affine linear maps is complete.
\end{thm}

Indeed, Valiant's original proof  (see more detail here~\cite{LiRe06})  implies a stronger theorem. Even restricting the support of each row to have at most a single variable in some coordinate, is general enough to be complete.

\begin{thm}
The class $R_{\rm sparse}$ of affine linear maps, such that each map is non-constant in at most a single variable from $\{ x_1, \dots x_d \}$, 
is complete.
\end{thm}

We now turn to define the polynomial maps we will be interested in, and for which we will be able to provide efficient deterministic algorithms. Some motivation for interest in these maps will be given in the next two subsections.

Consider the following class $R_2$. Here $d=n$. Every $p\in R_2$ is of the form $\x\mapsto (A-A^T)\x$, where $A$ is a rank-1 matrix. While this family may look very special, we note that the problem of graph rigidity in $\R^2$ (for which a polynomial time algorithm is known but far from trivial) is a very special case of PIT($R_2$).\footnote{Moreover, the same family of rank-2, skew symmetric matrices is featured in a very different PIT problem: determining the maximum rank of a subspace generated by given such matrices. A deterministic polynomial time solution for this problem is given by Lovasz' celebrated matroid parity algorithm~\cite{Lov80} (see also~\cite{LoPl09}, Theorem 11.1.2).}

\begin{thm}\label{R2}
PIT($R_2$) can be solved in deterministic polynomial time, over a field $\K$ with sufficiently large characteristic 
(more precisely, when $\char(\K)$ is larger than the number of rows of the input matrix or $\char(\K)=0$).
\end{thm}

This construction can be generalized as follows. Here we will generate PIT instances whose entries are {\em polynomials}, rather than linear functions of the variables.
For a $k$-dimensional tensor $A$ of size $n$, denote by $\hat A$ its ``anti-symmetric'' version, namely where for every entry $(i_1,\dots, i_k)$ we have 
$\hat A (i_1, \dots, i_k) = \sum_{\sigma \in S_k} \sgn (\sigma) A(i_{\sigma(1)}, \dots , i_{\sigma(k)})$.
Note that for $k=2$ we have $\hat A = A - A^T$. 

We now extend $R_2$, in which a matrix (namely a 2-dimensional tensor) acts on one vector of variables, to $R_k$, in which a $k$-dimensional tensor acts on $k-1$ vectors of variables.
Let $R_k$ denote the following class of (degree $k-1$) maps. Let $\x^1, \x^2, \dots , \x^{k-1}$ be $n$-vectors of independent variables, so altogether $\x = (\x^1, \x^2, \dots , \x^{k-1})$ is a vector of $(k-1)n$ variables. A $k$-tensor of size $n$ in each dimension acts on $\x$ simply with the $i$'th dimension acting on $\x^i$ for $i\in [k-1]$. The output of this action is a vector (along dimension $k$) of length $n$ of polynomials of degree $k-1$, each linear in $\x^i$ for all $i$. Define $R_k$ to be all maps defined by $\hat A$ for any rank-1 tensor $A$. Note that with this notation $R_2$ is precisely the class defined above.

Generalizing the above theorem we prove:
\begin{thm}\label{Rk}
For every $k<n$, PIT($R_k$) can be solved in deterministic polynomial time, over a field 
$\K$ with sufficiently large characteristic 
(more precisely, when $\char(\K)$ is larger than the number of rows of the input matrix or $\char(\K)=0$).
\end{thm}

\subsection{Graph Rigidity} 

The problem of graph rigidity arises from several motivations, originally, mechanical engineering (see \cite{Lam70}). 
Rigidity theory is a fast-growing area, and we refer the interested reader to \cite{SJS} for more background and recent approaches.
Graph rigidiy has several versions, 
we describe perhaps the most common one, {\em generic} rigidity. It is supposed to capture the structural rigidity of a ``bars and joints'' framework described by a graph. We will not be formal here as precise definitions can be found e.g. in \cite{AR1}. Here the relevant field for the geometric/physical interpretation is the Real numbers $\R$, and we use it in this subsection as in other papers on this problem (although the algebraic formulation is meaningful for every field $\K$). 

Let $G(V,E)$ be an undirected graph on $n$ vertices and $m$ edges. An {\em embedding} of $G$ in $\R^t$ is  a map $\phi: V \rightarrow \R^t$. An embedding of $G$ is called {\em rigid} if there is no perturbation of the vertex positions which preserves all edge lengths, other than the rigid motions of $\R^t$. The graph $G$ is called {\em rigid} if every {\em generic} embedding of $G$ is rigid (equivalently, if there exists an embedding of $G$ which is rigid, see~\cite{AR1}). The main question is to determine if a  given graph $G$ is rigid (and more generally, compute the dimension of the non-rigid motions of a generic embedding, in case $G$ is not rigid).

An extremely convenient formulation of the problem (as a PIT) is the following. Let $x_{v,j}$ be a set of variables indexed by $v\in V$ and $j\in [t]$. The intuition is that $(x_{v,1}, \dots , x_{v,t})$ are the coordinates of a generic embedding of the vertex $v$ in $\R^t$. Given $G$, construct a symbolic matrix $M_{G,t}$ of dimensions $m\times nt$, which may be viewed as a concatenation of $t$ matrices, one for each dimension $j\in [t]$.
Every row corresponds to an edge $\{u,v\}\in E$, and for each $j$, the column $u,j$ contains the entry $x_{u,j}-x_{v,j}$, whereas the column $v,j$ contains the the negation $x_{v,j}-x_{u,j}$.

It is not hard to prove that the rank (as usual, over $\R(x)$)
 of $M_{G,t}$ determines  if $G$ is rigid, and indeed the dimension of non-rigid motions (see \cite{AR1} for the details).
It is easy to see that for every graph $G$, the matrix $M_{G,2}$ is in the class $PIT(R_2)$ 
above. 
Indeed, let $e_1,\ldots, e_{2n}$ denote the standard basis vectors in $\R^{2n}$.
For some $u<v\in[n]$, put $a=e_u-e_v$ and $b=e_{n+u}-e_{n+v}$. Consider the matrix $A=A_{u,v}:=a^tb$. Then $(A-A^t)\x$, where $\x=(x_{21},\ldots,x_{2n}, x_{11},\ldots,x_{1n})$ is the $\{u,v\}$ row of $M_{G,2}$. 
Thus Theorem~\ref{R2} yields as a corollary a polynomial time algorithm
to determine whether a given graph $G$ is rigid in $\R^2$.
Such algorithms for rigidity in $\R^2$ are known (see \cite[Section 2.2]{Hen92} and references therein).
Note that the matrices $M_{G,t}$ make sense over any field $\K$, instead of $\R$, and Theorem~\ref{R2} in fact provides a deterministic polynomial time algorithm to compute the rank of these matrices   over any field $\K$ with large enough characteristic.

The symbolic matrix representation above shows that for every $t$, the problem of testing graph rigidity in $\R^t$ is in $BPP$, and it is a decades-old problem to whether it is also in $P$, even for the case $t=3$.

Lov\'asz and Yemini~\cite{LoYe} have developed an alternative approach for studying graph rigidity in the plane, which obtains a somewhat finer characterization of rigidity than Laman's. What is even more interesting is their method. They show that the matrices $M_{G,2}$ can actually be obtained in the following way. First, with every edge $\{u,v\}$ associate a certain $2$-dimensional subspace $f_{u,v} \subset \R^{2n}$. The intersection of this subspace $f_{u,v}$ with a {\em generic} hyperplane through the origin (of which the normal can be viewed essentially as the $2n$-vector of variables $x_{v,j}$) yields the $\{u,v\}$ row of $M_{G,2}$. In more detail, identify the vertices of $G$ with the set $V=[n]$, and let $e_1,\ldots,e_{2n}$ denote the standard basis in $\R^{2n}$. Define $f_{u,v}$ to be the subspace of $\R^{2n}$ spanned by the pair of vectors $e_u-e_v$ and $e_{n+u}-e_{n+v}$ (note that the definition of $f_{u,v}$ is symmetric in $u,v$). 
Let $h(\x)$ denote the subspace of $\R^{2n}$ orthogonal to the vector $\x= (y_1,\ldots,y_n,-x_1,\ldots,-x_n)$. It is not hard to verify (see \cite{LoYe} for the details) that $h(\x)\cap f_{u,v}$ is spanned by the $\{u,v\}$ row of $M_{G,2}$. Thus, for a generic  $\x$, we have
$$
\rk M_{G,2}=\dim \sp\{h(\x)\cap f_{u,v}\mid \{u,v\}\in E\}.
$$

Thus, the question of computing the rank of $M_{G,2}$ becomes the question of computing the dimension of the span of the resulting intersections (which here are simply lines) with a {\it generic} hyperplane. 
To analyze this, Lov\'asz and Yemini use a theory developed by Lov\'asz~\cite{Lov} which studies a similar problem for an arbitrary family of subspaces.
The relevant part of Lov\'asz's theory is introduced in the next subsection.

The idea of \cite{LoYe} can be applied also to rigidity in higher dimensions.
For simplicity of the presentation, let us consider only the case $t=3$. In this case we associate with each edge $\{u,v\}\in E$ a 3-dimensional subspace $g_{u,v}$ of $\R^{3n}$. Namely, the subspace spanned by the vectors $e_u-e_v$, $e_{n+u}-e_{n+v}$, $e_{2n+u}-e_{2n+v}$, where here $e_1,\ldots, e_{3n}$ stand for the standard basis of $\R^{3n}$.
Let $\x=(x_1,\ldots,x_n,y_1,\ldots,y_n,z_1,\ldots,z_n)$ and define $\tilde h(\x)$ to be the (codim 2) subspace of  $\R^{3n}$ orthogonal to the pair of vectors
$$(y_1,\ldots,y_n,-x_1,\ldots,-x_n,0,\ldots,0)$$
$$
(z_1,\ldots,z_n,0,\ldots,0,-x_1,\ldots,-x_n).
$$
It is not hard to verify that $\tilde h(\x)\cap f_{u,v}$ is one dimensional and spanned by the $\{u,v\}$ row of $M_{G,3}$. Thus, for a generic choice of $\x$, we have
$$
\rk M_{G,3}=\dim \sp\{\tilde h(\x)\cap f_{u,v}\mid \{u,v\}\in E\}.
$$
A crucial difference from the case $t=2$ is that here a generic choice of $\x$ does not yield a generic codim 2 subspace $\tilde h(\x)$ of $\R^{3n}$. From the perspective of this method and of our paper, this is ``the reason'' why rigidity in higher dimensions is more challenging.

\subsection{Subspaces and generic hyperplanes}

Let $F$ be a collection of subspaces in $\K^d$. Let $h$ be a generic hyperplane in $\K^d$, which without loss of generality can be taken to be all vectors perpendicular to $\x = (x_1, \dots x_d)$. For each subspace $f\in F$, let $f' = f\cap h$. Now consider the space spanned by the subspaces in $F':=\{f'\mid f\in F\}$ (note that the flats in $F'$ are functions of $\x$). The question is, what is the dimension of $\sp(F')$?

One of the major results of Lov\'asz' paper~\cite{Lov} is a  formula, called $\rho(F)$ (which we redefine in Section~\ref{sec:rho}), that determines this dimension for every family of subspaces, and for $\x$ satisfying a certain
``general position'' condition (see Definition~\ref{def:genp}). 
To show that a {\it generic} $\x$ satisfies Lov\'asz's general position condition  over any field (with large enough characteristic) is one main result of our paper (see Section~\ref{sec:generic}).
Note that this fact is mentioned (over the field $\R$) in \cite{Lov} with no proof. This fact is again mentioned\footnote{In Tanigawa~\cite{Tan12} an alternative general position condition is suggested, to supposedly correct a mistake in Lov\'asz's paper. However we find the counter example in \cite[footnote on p. 1416]{Tan12} false. We provide a full and detailed proof of Lov\'asz's formula in Section~\ref{sec:Lov}.}
 and applied, again with no proof, in Tanigawa~\cite{Tan12}. 
We see our paper as contributing to the completeness of these results.

When the subspaces $F$ are derived from a graph in the manner described above to generate the rigidity matrix, Lov\'asz and Yemini~\cite{LoYe} write the explicit special case of the formula $\rho(F)$, which yields an elegant characterization.
For the general case of an arbitrary family of subspaces $F$, the formula is given as the minimum, over all possible partitions of the family, of a certain easily computable function. As the number of partitions is exponential, there is no obvious efficient way of computing $\rho$. 
We have recently learned that the problem of computing $\rho$ is a special case of minimizing, over all partitions of a set $S$, the {\it Dilworth truncation} of a given submodular function $f$ defined over $S$; 
a strongly polynomial algorithm for this problem is given in Frank and Tardos~\cite[Chapters II.1 and IV.3]{FT88}. In our paper we introduce an alternative\footnote{Our algorithm seems different than the one in \cite{FT88}, as it does not use duality.} strongly polynomial algorithm for computing $\rho$, by reducing the original problem to a minimization problem of a certain submodular function. 
In fact, we prove our result to a more general quantity $\rho_c(F)$, introduced in Section~\ref{sec:rho}. (Note that $\rho(F)=\rho_1(F)$ is the quantity from \cite{Lov}.)
\begin{thm}\label{main}
There is a deterministic, strongly polynomial time algorithm to compute $\rho_c$ for every real number $c$.
\end{thm}

Closing this circle, we will also prove that the problem of computing $\rho_1$ is {\em equivalent} to PIT($R_2$).
This will yield Theorem~\ref{R2} as a corollary to Theorem~\ref{main}.


\subsection{Related works and applications}
We see our result as a step towards better understanding of the algorithmic aspects of the notions and formulas introduced in Lov\'azs \cite{Lov} and their applications. 

Let us mention one related concept studied in Lov\'asz~\cite{Lov} and discuss follow-up work by Tanigawa~\cite{Tan12}, which is related to Theorem~\ref{Lovthm} proved in this paper. It would be interesting to find efficient algorithms for the natural computational problem at hand. 
The reader may skip this subsection at first reading.

Let $F$ be a finite family of subspace in $\K^d$ (where $\K$ is a field of characteristic $0$). 
Let $X=\{x_f\mid f\in F\}$ be a collection of points in $\K^d$ such that $x_f\in f$ for each $f\in F$.
The set $X$ is said to be in {\it general position} with respect to $F$ if, for every $f\in F$ fixed, the following holds: Any subspace spanned by members of $F$ and points of $X\setminus\{x_f\}$ containing $x_f$ must contain the whole flat $f$.
Lov\'asz shows that there exists a choice of a set $X$ in general position with respect to any given family $F$. He then proves the following formula:
\begin{thm}[{\bf Lov\'asz~\cite{Lov}}]\label{lovpoints}
Let $F$ be a finite family of subspace in $\K^d$, and let $X=\{x_f\mid f\in F\}$ be in general position with respect to $F$. Then 
$$
\rk(\sp X)=\min_{G\subseteq F} \left\{\rk(\sp\bigcup G)+|F\setminus G|\right\}
$$
\end{thm}

An interesting application of Theorem~\ref{lovpoints} to the body-rod-bar rigidity problem is obtained by Tanigawa~\cite{Tan12}.
A {\it body-rod-bar framework} in $\R^d$ is
defined as a structure consisting of $d$-dimensional subspaces (bodies) and $(d-2)$-dimensional flats (rods) mutually linked by one-dimensional lines (bars). (The term ``rod'' is appropriate for $d=3$.)
More formally, a $d$-dimensional body-rod-bar-framework is a triple $(G,q,r)$, where 
$G=(V=B\cup R, E)$ is a graph, $r:R\to {\rm Gr}(d-1,\R^{d+1})\subset \mathbb{P}(\bigwedge^{d-1} (\R^{d+1}))$ is the {\it rod-configuration} mapping a vertex $v\in R$ to a $(d-1)$-dimensional subspace $r_v$ of $\R^{d+1}$, and $q:E\to {\rm Gr}(2, \R^{d+1})\subset \mathbb{P}(\bigwedge^2 (\R^{d+1}))$ is the {\it bar-configuration} mapping an edge $e\in E$ to a 2-dimensional subspace $q_e$ in $\R^{d+1}$, such that
$$
\text{$q_e$ and $r_v$ have a nonzero intersection,
whenever $v\in R$ is a vertex of $e$;}
$$
equivalently, 
$$
\text{$q_e\cdot r_v=0$, 
whenever $v\in R$ is a vertex of $e$,}
$$
where here the dot product should be interpreted appropriately (see \cite{Tan12} for the details).
Assume also that $r(u)\neq r(v)$ for every $u\neq v\in R$.

An {\it infinitesimal motion} of $(G,q,r)$ is a mapping $m:B\cup R\to \bigwedge^{d-1}(\R^{d+1})$ such that 
\begin{equation}\label{indep}
q_e\cdot (m(u)-m(v))=0, ~~\text{for every $e=\{u,v\}\in E$.}
\end{equation}
 An infinitesimal motion $m$ is called {\it trivial} if either $m(u)=m(v)$ for all $u,v\in V$, or if, for some fixed $v_0\in V$ we have $m(v_0)=r_{v_0}$ and $m(v)=0$ for every $v\in V\setminus \{v_0\}$.
Finally, a framework $(G,q,r)$ is called {\it infinitesimally rigid} if every infinitesimal motion is trivial.

The body-rod-bar problem gives rise to a matroid ${\rm BR}(G, q, r)$ defined on the edge set $E$ whose rank is the maximum size of independent linear equations in \eqref{indep} (for unknown m). From the definition, $(G, q, r)$ is infinitesimally rigid if and only if the rank of ${\rm BR}(G, q, r)$ is 
$
\tbinom{d+1}{2}|V|- (\tbinom{d+1}{2}+ |R|)$.

\begin{thm}[{\bf Tanigawa~\cite[Corollary 4.13]{Tan12}}]\label{tanig}
Let $G=(B\cup R, E)$ and suppose $d\ge 3$.
Then, for almost all bar-configurations $q$ and almost all rod-configurations $r$ we have
$$
\rk(E)=\min_{\Pi=\{F_0,\ldots,F_k\}}\left\{|F_0|+ \sum_{i=1}^k\left(\tbinom{d+1}{2}(V(F_i)-\tbinom{d+1}{2}-R(F_i)\right)\right\},
$$
where the minimum is taken over all partitions $\Pi$ of $E$.
\end{thm}

Tanigawa's proof is a nice combination of Theorem~\ref{lovpoints} with the other  result of Lov\'asz mentioned in the introduction, cited below  as Theorem~\ref{Lovthm}. Briefly, the first (simpler) step in the proof is to reduce the problem to the form of Theorem~\ref{lovpoints}. That is, a family of flats $F$ is introduced, and the question becomes to find the rank of a generic set of points $X=\{x_f\mid f\in F\}$.
 The family $F$ resulted from the reduction can be described as follow: Each edge $e=\{u,v\}$ of $G$ is associated with some fixed subspace $f_e$ in $\left(\mathbb{P}(\bigwedge^2(\R^{d+1}))\right)^{|V|}$. Then 
 $F=\{f_e\cap h(u)\cap h(v)\mid e=\{u,v\}\in E\}$, where $h_r(u),h_r(v)$ are subspaces depending on the choice of rod configuration $r$.
Since $r$ is taken generically, this imposes some genericity on the subspaces $h_r(v)$, but they are not exactly generic. The proof is then complete by proving a relaxed version of Theorem~\ref{Lovthm}, and adding the subspaces $h_r(v)$ one after the other.

For more recent applications of \cite{Lov,LoYe} see Tanigawa~\cite{Tan12, Tan15}.

\subsection{Organization of this paper}

In Section~\ref{sec:rho} we introduce the function $\rho_c(F)$, which is the main object of this study.
The rest of the paper has two separate parts. The first, in Sections \ref{sec:proprho} and \ref{sec:alg}, describes the algorithm to compute $\rho_c$. In Section~\ref{sec:proprho}, we present and prove properties of the function $\rho_c$. Using these properties we describe, in Section~\ref{sec:alg},  a deterministic strongly polynomial time algorithm that computes $\rho_c$ over every field via submodular optimization.  Note that, as there is an alternative algorithm \cite{FT88} in the literature to efficiently compute functions like $\rho_c$, this part can be skipped.

The second part, in Sections \ref{sec:Lov}, \ref{sec:reduction}, and \ref{sec:generic}, describes the genericity proof of $\rho$.
In Section~\ref{sec:Lov}, we state (and reprove) the result of Lov\'asz~\cite{Lov} above, relating $\rho_1$ to the intersection of $F$ with a hyperplane in  ``general position''. A similar relation is obtained for $\rho_c$, for an integer $c>0$ (see Theorem~\ref{codimk}).
In Section~\ref{sec:reduction}, we develop an explicit representation of a basis of the family $F'$ resulting from this  intersection, which give rise to the symbolic matrices PIT($R_2$) (and PIT($R_k$)). Using this, we prove in Section~\ref{sec:generic} that most hyperplanes (and more generally, subspaces) satisfy the ``general position'' definition of Lov\'asz, thus expressing the rank of a these symbolic matrices as appropriate $\rho(F)$. Using the algorithm above we can now compute these ranks deterministically and efficiently. This last section is the only one in which the size of the field $\K$ is important. 
%


\section{Subspaces, partitions, and the function $\rho_c$}\label{sec:rho}
We introduce the main objects of this study: Families of subspaces, their partitions, and the optimization problem we solve in this paper.
We consider linear subspaces $f$ of $\K^d$. 
Let $d(f)$ denote the dimension of a subspace $f$.
For a family $F$ of subspaces, we write $\sp F:=\sp\bigcup_{f\in F}f$ and
$$
d(F):=d(\sp F).
$$ 
A {\it partition} of $F$ is a set $\Pi=\{P_1,\ldots,P_t\}$ of nonempty, pairwise disjoint subfamilies of $F$, such that $F=\bigcup_{i=1}^tP_i$.
For a partition $\Pi$ of $F$ and a family of subspaces $G$, we define the {\it restriction} of $\Pi$ to $G$ by
\begin{equation}\label{restpart}
\Pi\cap G:=\{P\cap G\mid P\in \Pi,\;\; P\cap G\neq\emptyset\}.
\end{equation}
If $G\subset F$, then $\Pi\cap G$ forms a partition of $G$.

Lov\'asz~\cite{Lov} defined the following key function $\rho$ of a family of subspaces, whose meaning will be revealed in Section~\ref{sec:Lov}. We actually generalize his definition to a family of functions $\rho_c$, for every $c>0$ (his $\rho$ is our $\rho_1$ for $c=1$). Computing $\rho_c(F)$ in deterministic polynomial time given $F$, in Section~\ref{sec:alg}, will be the key to our derandomization results.

Fix a constant $c>0$. Let $F$ be a finite family of subspaces in $\K^d$. 
For a partition $\Pi$ of $F$,
we define 
$$
\rho_c(F,\Pi):=\sum_{P\in \Pi}(d(P)-c).
$$

\begin{equation}\label{rhodef}
\rho_c(F):=\min_{\Pi}\rho_c(F,\Pi),
\end{equation}
where the minimum is taken over all partitions $\Pi$ of $F$. 
\begin{defin}
We say that $\Pi$ is a {\em minimal} partition of $F$, with respect to the constant $c>0$, if $\Pi$ attains $\rho_c(F)$ and has the smallest possible number of parts.
\end{defin}
\noindent{\it Remark.} In Corollary~\ref{unique} we prove that, fixing $c>0$, a minimal partition $\Pi$ of a family $F$ with respect to $c$ is unique.


\paragraph{Notation.} We will use small letters $f,g,h$ to denote subspaces in $\K^d$, capital letters $F,G,P,Q$ to denote families of subspaces, and $\Pi$ to denote partitions of a certain family $F$ of subspaces. Note that the elements of a partition $\Pi$ are themselves families of subspaces.


\section{Properties of minimal partitions}\label{sec:proprho}

In this and the next section we develop our algorithm in a fully self-contained manner. As mentioned in the introduction, the reader may skip these sections and apply the algorithm of \cite{FT88} as a black box. 
In this section, we introduce some properties of minimal partitions, to be used in our algorithm. We find these properties interesting in their own right, but some may be  known, indeed in more generality, for submosular functions.
\subsection{Main technical lemma}
We start with the following main technical lemma of this section.
\begin{lem}\label{mainlem}
Let $F,G$ be families of subspaces in $\K^d$ with minimal partitions $\Pi_F,\Pi_G$, respectively. 
Assume that $Q\in \Pi_G$ and $Q\subset F$. Then $Q$ is contained in one of the parts of $\Pi_F$.
\end{lem}
For the proof, the idea is to show that if, when considering a minimal partition for $F$, it ``pays off'' to put the elements of $Q$ together, then it still ``pays off'' (or at least, harmless) to put these elements together, when this time considering a minimal partition for $G$.
\begin{proof}
Consider the restriction $\Pi':=\Pi_F\cap Q$ of $\Pi_F$ to $Q$  (as defined in \eqref{restpart}).
By assumption, $Q\subset F$, and thus $\Pi'$ forms a partition of $Q$.

Our assumption that $Q\in\Pi_G$, and recalling that $\Pi_G$ forms a minimal partition of $G$, implies that 
\begin{equation}\label{eq1}
\sum_{P\in\Pi'}(d(P) -c)\ge d(Q)-c.
\end{equation}

Fixing some arbitrary order on the elements of $\Pi'$, we write
$$
\Pi'=(P_1',\ldots,P_t'),
$$ 
where $P_i':=P_i\cap Q$ is non-empty and $P_1,\ldots, P_t\in\Pi_F$ are distinct.
Set $V_0':=\{0\}$. For each $1\le i\le t$, define 
$$
V_i':=\sp\left(\bigcup_{j=1}^i P_j'\right)
$$ 
and 
put $r_i':=d(V_i')-d(V'_{i-1})$ and $s_i':=d(P_i')-r_i'$. 
Note that 
$$
d(Q)=\sum_{i=1}^tr_i'
$$ and that 
\begin{equation}\label{si'}
s_i'=d((\sp P_i')\cap V_{i-1}').
\end{equation}
With this notation, \eqref{eq1} can be rewritten as
$$
\sum_{i=1}^t(r_i'+s_i') -tc\ge \sum_{i=1}^tr_i' -c
$$
which implies 
\begin{equation}\label{sumsi'}
\sum_{i=1}^ts_i' \ge c(t -1).
\end{equation}

Next, we define
$$
V_i:=\sp\left(\bigcup_{j=1}^i P_j\right)
$$ 
and put 
$r_i:=d(V_i)-d(V_{i-1})$ and $s_i:=d(P_i)-r_i$. 
Similar to above, we have
$$
d\left(\bigcup_{i=1}^tP_i\right)=\sum_{i=1}^tr_i
$$ and  
\begin{equation}\label{si}
s_i=d((\sp P_i)\cap V_{i-1}).
\end{equation}
We claim that 
\begin{equation}\label{eq2}
\sum_{i=1}^t(d(P_i) -c)\ge d\left(\bigcup_{i=1}^tP_i\right)-c.
\end{equation}
Indeed, the inequality \eqref{eq2} holds if and only if
$$
\sum_{i=1}^t(r_i+s_i) -tc\ge \sum_{i=1}^tr_i -c
$$
which holds if and only if 
\begin{equation}\label{sumsi}
\sum_{i=1}^ts_i \ge c(t -1).
\end{equation}
To prove the last inequality, notice that 
$V_i'\subset V_i$ and $\sp P_i'\subset\sp P_i$, for every $i$. Thus
$$
d((\sp P_i')\cap V_{i-1}')\le d((\sp P_i)\cap V_{i-1}).
$$ 
Hence, by \eqref{si'} and \eqref{si}, we get $s_i'\le s_i$. 
This fact combined with the inequality \eqref{sumsi'} implies \eqref{sumsi} and hence also \eqref{eq2}.
Since $\Pi_F$ is assumed to be minimal for $F$, we conclude that $t=1$ and
$Q\subset P_1$. This completes the proof.
\end{proof}

\subsection{Uniqueness of minimal partitions}

We prove uniqueness of minimal partitions.
\begin{cor}[{\bf Uniqueness}]
\label{unique}
Let $F$ be a family of subspaces in $\K^d$ and let $\Pi_1,\Pi_2$ be minimal partitions of $F$. Then $\Pi_1=\Pi_2$.
\end{cor}
\begin{proof}
Let $\sim_1,\sim_2$ denote the equivalence relations on $F$ induced by the partitions $\Pi_1, \Pi_2$, respectively. Let $f,g\in F$ and assume that $f\sim_1 g$. That is $f,g\in Q$, for some $Q\in \Pi_1$.
Applying Lemma~\ref{mainlem} (with $F$, $G:=F$, and $Q$), we get that $Q$ is contained in one of the parts in $\Pi_2$.
Thus $f\sim_2 g$.
By symmetry, we conclude that $f\sim_1 g$ if and only if $f\sim_2 g$.
Thus $\Pi_1=\Pi_2$, as claimed.
\end{proof}

\begin{defin} Fix $c>0$. Define $\Pi^*(F)$ to be {\it the} minimal partition of a family of subspaces $F$ (with respect to $c$).
\end{defin}

\subsection{Monotonicity properties}
We prove the following ``monotonicity" property of minimal partitions.
\begin{cor}[{\bf Monotonicity}]\label{monotone}
Let $F,G$ be families of subspaces in $\K^d$ and assume that $G\subset F$.
Then $\Pi^*(G)$ is a refinement of $\Pi^*(F)\cap G$.
\end{cor}
\begin{proof}
Apply Lemma~\ref{mainlem} to the families $F$ and $G$.
\end{proof}

The following is another type of monotonicity property.
\begin{lem}
Let $F=\{f_1,\ldots,f_n\}$ be a family of $n$ subspaces in $\K^d$.
Let $f_i\subset f_i'$, for every $i=1,\ldots,n$, and consider
$
F':=\{f_1',\ldots,f_n'\}.
$  
For a partition $\Pi$ of $F$, let $\Pi'$ denote the partition of $F'$ induced by $\Pi$, replacing each $f_i$ by the corresponding $f_i'$.
Then $(\Pi^*(F))'$ is a refinement of $\Pi^*(F')$.
\end{lem}
\begin{proof}
Let $P\in\Pi^*(F)$ and assume without loss of generality that $P=\{f_1,\ldots,f_m\}$, for some $m\le n$. It is easy to see, applying Lemma~\ref{mainlem}, that $\Pi^*(P)=\{P\}$.

Put $P':=\{f_1',\ldots,f_m'\}$. We claim that $\Pi^*(P')=\{P'\}$.
First note that it suffices to prove the claim for the special case where $f_1\subset f_1'$ and $f_i=f_i'$, for $i=2,\ldots,m$, and then apply the same argument repeatedly to each $i$. To prove the calim for the special case, consider the family $Q=\{f_1,f_1'\}$. It is easy to see, by definition, that $\Pi^*(Q)=\{Q\}$. By Lemma~\ref{mainlem}, $Q$ is contained in a part of $\Pi^*(G)$, for every family of subspaces $G$ that contains $Q$. Moreover, since $f_1\cup f_1'\subset f_1'$, we have
$$
\rho_c(G)=\rho_c(G\setminus\{f_1\})~~\text{and}~~\Pi^*(G\setminus\{f_1\})=\Pi^*(G)\cap(G\setminus \{f_1\})
$$
for every such $G$ (this follows directly from the definition of $\rho_c$ and of $\Pi^*$).

Define $G:=\{f_1,f_1',f_2,\ldots,f_m\}$. By what has just been argued, we have
\begin{equation}\label{eqH}
\Pi^*(P')=\Pi^*(G)\cap P'.\end{equation}
Since $P,Q\subset G$, and applying Lemma~\ref{mainlem}, we get that each of $P$ and $Q$ is contained in a part of $\Pi^*(G)$. But $P\cap Q\neq\emptyset$, thus the set $P\cup Q$  must be contained in a part of $\Pi^*(G)$. Noting that $P\cup Q=G$, this implies that $\Pi^*(G)=\{G\}$. Combined with \eqref{eqH}, this proves $\Pi^*(P')=P'$, as claimed.

Applying Lemma~\ref{mainlem} to the families $F'$, $P'$, and with $P'\in\Pi^*(P')$, we conclude that $P'$ is contained in one of the parts of $\Pi^*(F')$. Since this is true for every $P\in \Pi^*(F)$, the lemma follows.
\end{proof}

\subsection{The family $\hat F$}
Let $F$ be a family of subspaces in $\K^d$. We show that, in some sense, $F$ can be replaced by a simpler family $\hat F$ defined next.  With each  $P\in \Pi^*(F)$  associate the subspace $f_P:=\sp P$. 
Then define the family
$$
\hat F:=\{f_P\mid P\in \Pi^*(F)\}.
$$
Note that for $P\neq P'$ we have $f_P\neq f_{P'}$; otherwise, taking $P\cup P'$ yields a partition of $F$ with strictly less parts and with smaller or equal value of $\rho_c$, contradicting the minimality of $\Pi^*(F)$.

The family $F$ can be replaced by $\hat F$ in the sense of Lemma~\ref{hatPi}, and $\hat F$ is simpler in the sense of Lemma~\ref{parthatF}.
\begin{lem}\label{hatPi}
Let $F,G$ be families of subspaces in $\K^d$.
Then
$$
\rho_c(F\cup G)=\rho_c(\hat F\cup G) ~~\text{and}~~\Pi^*(F\cup G)\simeq\Pi^*(\hat F\cup G).
$$
\end{lem}
\noindent By the sign $\simeq$ we mean that the identity holds after identifying the partiton $\Pi^*(\hat F\cup G)$ of $\hat F\cup G$ with the partition of $F\cup G$ naturally induced by it. Concretely, 
the lemma asserts that
$$\Pi^*(F\cup G)=\{(\bigcup_{f_P\in \hat Q}P)\cup (G\cap \hat Q)~\mid~ \hat Q\in\Pi^*(\hat F\cup G)\}.
$$
\begin{proof}
In the proof we often abuse notation and regard a partition of $\hat F\cup G$ as a one of $F\cup G$, as explained after the statement of the lemma.
Let  $\Pi^*$ be the partition of $F\cup G$ induced by $\Pi^*(\hat F\cup G)$, given by
$$
\Pi^*=\Big\{(\bigcup_{f_P\in \hat Q}P)\cup (G\cap \hat Q)~\mid~ \hat Q\in\Pi^*(\hat F\cup G)\Big\}.
$$
We have $|\Pi^*|=|\Pi^*(\hat F\cup G)|$ and 
$$
\rho_c(F\cup G,\Pi^*)=\rho_c(\hat F\cup G,\Pi^*(\hat F\cup G)).
$$ 
Thus
$$
\rho_c(F\cup G)\le \rho_c(\hat F\cup G).
$$

To prove the inverse inequality, apply Lemma~\ref{mainlem} to the families $F$ and $F\cup G$.
It follows that, for every $P\in \Pi^*(F)$, there exists $Q\in \Pi^*(F\cup G)$ such that 
$P\subset Q$.
This means that $\Pi^*(F\cup G)$ induces a well-defined partition $\hat\Pi^*$ of $\hat F\cup G$ with $|\Pi^*(F\cup G)|=|\hat\Pi^*|$ and
\begin{equation}
\rho_c(F\cup G,\Pi^*(F\cup G))= \rho_c(\hat F\cup G,\hat\Pi^*).
\end{equation}
Concretely, $\hat \Pi^*$ is given by
$$
\hat \Pi^*:=\{\hat Q\mid Q\in \Pi^*(F\cup G)\},$$
where 
$$
\hat Q:=\left\{f_P\mid P\subset Q, P\in \Pi^*(F)\right\}\cup
(Q\cap G).
$$
We have
\begin{align*}
\rho_c(F\cup G)
&= \rho_c(F\cup G,\Pi^*(F\cup G))\\
&=\rho_c(\hat F\cup G,\hat\Pi^*)\\
&\ge \rho_c(\hat F\cup G).
\end{align*}
This proves that $\rho_c(F\cup G)=\rho_c(\hat F\cup G)$.

Next, we claim that $|\Pi^*(F\cup G)|=|\Pi^*(\hat F\cup G)|$.
Indeed, by our argument above, the partition $\hat\Pi^*$ of $\hat F\cup G$ satisfies 
$$
\rho_c(\hat F\cup G,\hat\Pi^*)=\rho_c(\hat F\cup G)~~\text{and}~~|\hat\Pi^*|=|\Pi^*(F\cup G)|.
$$
Since $\Pi^*(\hat F\cup G)$ is taken to be the smallest that attains $\rho_c(\hat F\cup G)$, we get
$$|\Pi^*(\hat F\cup G)|\le |\Pi^*(F\cup G)|.$$
Similarly, by our argument above, the partition $\Pi^*$ of $F\cup G$ satisfies 
$$
\rho_c(F\cup G,\Pi^*)=\rho_c(F\cup G)~~\text{and}~~|\Pi^*|=|\Pi^*(\hat F\cup G)|.
$$
Thus, 
$$|\Pi^*(F\cup G)|\le |\Pi^*(\hat F\cup G)|.$$
This proves the claim.

By the uniqueness of minimal partition (see Corollary~\ref{unique}), we conclude that
$$
\Pi^*(\hat F\cup G)=\hat \Pi^*~~\text{and}~~\Pi^*(F\cup G)=\Pi^*.
$$
This completes the proof of the lemma.
\end{proof}

\begin{lem}\label{parthatF}
Let $F$ be a family of subspaces in $\K^d$.
Then 
$$\Pi^*(\hat F)=\{\{\hat f\}\mid \hat f\in\hat F\}.$$
\end{lem}
\begin{proof}
Apply Lemma~\ref{hatPi} with $G=\emptyset$.
\end{proof}

We introduce one more simple property that we need.
\begin{lem}\label{hatassoc}
$
\widehat{F\cup G}=\widehat{\widehat F\cup G}.
$\end{lem}
\begin{proof}
By Lemma~\ref{hatPi}, $\Pi^*(F\cup G)=\Pi^*(\hat F\cup G)$. The assertion then easily follows.
\end{proof}

\section{An algorithm for computing $\rho_c(F)$}\label{sec:algorithm}
\label{sec:alg}
In this section we prove Theorem~\ref{main}. That is, we introduce an algorithm to compute $\rho_c(F)$, for any number $c$ and a given family $F$ of $n$ subspaces in $\K^d$, with polynomial running time in $n$ (and in $d$).
While we designed our algorithm for the class of functions $\rho_c$, it clearly works for a wider class of submodular functions. As it is different than the one in \cite{FT88}, we feel it would be interesting to explore its generality.
Note that the problem is trivial for $c\le 0$, which is why we consider only $c>0$.

As mentioned in the introduction, the problem of computing $\rho_c$ turns out to be an instance of a more general problem to which a strongly polynomial time algorithm is already known~\cite{FT88}. In more detail,  the {\it Dilworth truncation} of a set function $b':2^S\to \R\cup \{\infty\}$ is defined as the function
$$
b(X)=\min_{\Pi}\sum_{P\in\Pi} b'(P),
$$ 
where the minimum is taken over all partitions $\Pi$ of $X$. 

\begin{thm}[{\bf Frank and Tardos~\cite[IV.3]{FT88}}]
Let $b':2^S\to \R\cup \{\infty\}$ be a submodular set function. Suppose that a minimizing oracle for $b'$ is available.
Then $b(S)$ can be computed in a strongly polynomial time.
The algorithm also constructs a partition $\Pi$ of $S$ for which 
$b(S)=\sum_{P\in\Pi} b'(P)$.
\end{thm}

\noindent {\it Remark.} In \cite{FT88}, a more general result is proved.

\subsection{High-level description of the algorithm for $\rho_c$}\label{algdesc}
The input to the algorithm is a number $c$ and  a family of subspaces  $F=\{f_1,\ldots,f_n\}$  in $\K^d$
Write $F_i:=\{f_1,\ldots,f_i\}$. 
The high-level scheme of the algorithm is the following:
\renewcommand{\labelenumii}{\theenumii}
\renewcommand{\theenumii}{\theenumi.\arabic{enumii}.}
\begin{enumerate}
\item $\hat F_1\leftarrow \{f_1\}$.
\item {\bf For} $i\leftarrow$ $2$ {\bf to} $n$ 
\begin{enumerate}
\item $\Pi\leftarrow$ {\bf Compute }$\Pi^*(\hat F_{i-1}\cup\{f_i\})$
\item $\hat F_i\leftarrow\{\sp(P)\mid P\in\Pi\}$
\end{enumerate}
\item {\bf Return} $\sum_{\hat f\in \hat F_n}(d(\hat f)-c)$
\end{enumerate}
The heart of the algorithm is of course the missing description of Step 2.1, which computes, in the $i$th iteration, the minimal partition of the family $\widehat F_{i-1}\cup \{f_i\}$ with respect to $\rho$. 
\begin{lem}\label{lem:step2.1}
The computation in Step 2.1 can be done in strongly-polynomial time.
\end{lem}
Recall that the minimal partition of $\hat F_{i-1}$ is the partition into singletons, by Lemma~\ref{parthatF}. So in this step we compute the effect on this partition of inserting one new subspace. We explain how to do so efficiently and prove Lemma~\ref{lem:step2.1} in Section~\ref{sec:minpart} below.
To describe and analyze step 2.1, we first need to recall submodular functions and optimization, which we do in Section~\ref{sec:submod}. The proof of the lemma is then given in Section~\ref{sec:minpart}.

We are now ready to prove Theorem~\ref{main}, assuming that Lemma~\ref{lem:step2.1} is true.
\begin{proof}[Proof of Theorem~\ref{main}]\noindent{\it Correctness of the algorithm.}
By Lemma~\ref{hatassoc}, we have
$$\widehat F_i=\widehat{\widehat F_{i-1}\cup\{f_i\}}.$$
Thus the computation of $\hat F_i$ in Step 2.2 is correct. 
In view of Lemmas~\ref{hatPi} and \ref{parthatF}, the algorithm's output is $\rho_c(F)$, as needed. 

\noindent{\it Running time of the algorithm.}
We represent a $k$-dimensional subspace $f$ in $\K^d$ by a $k\times d$ matrix whose rows form a basis for $f$.
The dimension $d(f)$ of a subspace $f$ is just the number of rows in the matrix representing the subspace, and hence can be computed in a constant time.
Let $P$ be a family of subspaces in $\K^d$. To compute $\sp(P)$, we take the union of the rows of the matrices in $P$ (representing subspaces) and apply Gauss elimination (using row operations only). If $P$ has $n$ subspaces, we will need to apply Gauss elimination to a matrix of dimensions at most $(nd)\times d$. The nonzero rows in the matrix received by this process will form a basis for $\sp(P)$.

Now let $F$ be a family of $n$ subspaces in $\K^d$. Cleary, each line in the above description of the algorithm, when applied to $F$, is called at most $n$ times. In each step, excluding Step 2.1, we are required to compute at most $n$ times one of the operations just described (finding dimension or span) or simple operations such as addition. In view of Lemma~\ref{lem:step2.1}, the proof is complete.
\end{proof}

\subsection{A submodular set function}\label{sec:submod}
Recall that a function $s$ defined on the
collection of subsets of a finite set $A$  is called {\it submodular} if 
$$
s(X)+s(Y) \ge  s(X\cup Y)+ s(X\cap Y) 
$$
for all $X,Y\subset A$.

The following is proved by Schrijver in~\cite{Schr}.
\begin{thm}[{\bf Schrijver~\cite{Schr}}]\label{Schr}
There exists a strongly polynomial-time algorithm minimizing a submodular function $s$, where $s$ is given by an oracle. The number of oracle calls is bounded by a polynomial in the size of the underlying set. The algorithm also finds a minimizer $X^*$ of $s$.
\end{thm}


In this section we consider a set function defined as follows. Let  $F$ be a family of subspaces in $\K^d$ and let $g\subset\K^d$ be a subspace not in $F$. 
Fix $c>0$.
Define $r_{F,g,c}:2^F\to \K$  by  
$$
r_{F,g,c}(X):=d\left(X\cup \{g\}\right)-c+\sum_{f\in\overline X}(d(f)-c),
$$
where $\overline{X}:=F\setminus X$.
We then put
$$
r_{F,g,c}^*:=\min_{X\subset F}r_{F,g,c}(X)
$$
and we let $X_{F,g,c}^*$ denote a subset $X\subset F$ that attains $r_{F,g,c}^*$.

We show that $r_{F,g,c}$ is submodular.
\begin{lem}
Let $F$ and $g$ and $c$ be as above. Then  $r_{F,g,c}$
 is submodular.
\end{lem}
\begin{proof}
To simplify the notation, and as $F,g,c$ are fixed, we write for short $r=r_{F,g,c}$. Let $X,Y\subset F$. 
We need to show 
$$
r(X)+r(Y)\ge r(X\cup Y)+r(X\cap Y).
$$
Put $f_X:=\sp(X\cup\{g\})$.
By definition, we have
\begin{align*}
r(X)+r(Y)
&=d(X\cup\{g\})+d(Y\cup\{g\})
+\sum_{f\in \bar X}d(f)+\sum_{f\in\bar Y}d(f)-c|\bar X|-c|\bar Y|-2c\\
&=d(f_X)+d(f_Y)
+\sum_{f\in \bar X}d(f)+\sum_{f\in\bar Y}d(f)-c|\bar X|-c|\bar Y|-2c.
\end{align*}
By basic linear algebra, we have the identity
$$
d(f_X)+d(f_Y)=d(\sp(f_X\cup f_Y))+d(f_X\cap f_Y).
$$
Thus the last equality, after some rearranging, is
\begin{align*}
r(X)&+r(Y)=
\\
&\Big(d(\sp(f_X\cup f_Y))-c+\sum_{f\in \bar X\cap\bar Y}d(f)-c|\bar X\cap \bar Y|\Big)
+\Big(d(f_X\cap f_Y)-c+
 \sum_{f\in \bar X\cup\bar Y}d(f)-c|\bar X\cup \bar Y|\Big)\end{align*}
Noting that $\sp(f_X\cup f_Y)=\sp (f_{X\cup Y})$ and 
that 
$\sp(f_X\cap f_Y)\supset\sp (f_{X\cap Y})$, we get
\begin{align*}
r(X)+r(Y)
&\ge\left(d(f_{X\cup Y})-c+\sum_{f\in \overline {X\cup Y}}d(f)-c|\overline {X\cup Y}|\right)
+\left(d(f_{X\cap Y})-c+\sum_{f\in \overline {X\cap Y}}d(f)-c|\overline {X\cap Y}|\right)\\
&=r(X\cup Y)+r(X\cap Y).
\end{align*}
This proves the lemma.
\end{proof}

\subsection{Inserting one subspace}\label{sec:minpart}

We are now ready to describe in detail Step 2.1 which computes $\widehat F_i$ given $\widehat F_{i-1}$ and $f_i$.
More precisely, we describe a subroutine that receives as an input a family $F$ with $F=\widehat F$ and a subspace $g$, and outputs $\Pi^*(F\cup \{g\})$.

We will need the following observation.
\begin{lem}\label{insert1}
Let $G=F\cup \{g\}$ be a family of subspaces in $\K^d$. 
Let $Q_g\in\Pi^*(G)$ be the part that contains the subspace $g$.
Then 
$$
\Pi^*(G)\setminus\{Q_g\}\subset\Pi^*(F).
$$
\end{lem}
\begin{proof}
For every $Q\in\Pi^*(G)\setminus\{Q_g\}$, we have $Q\subset F$.
By Lemma~\ref{mainlem}, there exists $P\in\Pi^*(F)$ such that $Q\subset P$.
Clearly, we also have $P\subset G$. Applying Lemma~\ref{mainlem} once again, we get that also $P\subset Q$. Thus, $P=Q$ which means that $Q\in \Pi^*(F)$. 
\end{proof}

\begin{cor}\label{cor:compr}
Let $F$ be a family of $n$ subspaces in $\K^d$ with $\hat F=F$ and let $g$ be another subspace in $\K^d$.
Then $\rho_c(F\cup \{g\})=r_{F,g,c}^*$ and 
$$
\Pi^*(F\cup\{g\})=\{X_{F,g,c}^*\cup \{g\}\}\cup\{\{f\}\mid f\in F\setminus X_{F,g,c}^*\},
$$
where $X_{F,g,c}^*$ and $r_{F,g,c}^*$ are as defined in Section~\ref{sec:submod}.\end{cor}
\begin{proof}
This follows from the definitions of $\rho_c$ and $r_{F,g,c}^*$, combined with Lemma~\ref{insert1}.
\end{proof}

\begin{proof}[Proof of Lemma~\ref{lem:step2.1}.]
Combinig Corollary~\ref{cor:compr} with Theorem~\ref{Schr}, we get that the computation in Step 2.1 can be done in strongly-polynomial time. 
\end{proof}

\section{Intersecting subspaces with a hyperplane}\label{sec:Lov}

In this section we state (and reprove) a result of Lov\'asz~\cite{Lov}, which explains the source of the function $\rho$ (more precisely, taking $\rho_c$ with  $c=1$) as the dimension of the intersections of a family of subspaces with a  hyperplane in ``general position''. This connection has been used by Lov\'asz to study certain questions about matroids in \cite{Lov}, and by Lov\'asz and Yemini in \cite{LoYe} to study rigid structures in $\R^2$. We extend Lov\'asz' treatment to arbitrary fields $\K$.

In Theorem~\ref{codimk} below, we further extend Lov\'asz's result, in a straightforward manner, to apply to the  intersection of a family of subspaces with an arbitrary  subspace  (of any co-dimension) in ``general position'', instead of only a (co-dimension 1) hyperplane.

Lov\'asz~\cite{Lov} uses a very specific notion of genericity, which he calls {\em general position} defined below, and shows that $\rho$ correctly computes the dimension of the intersection when the hyperplane is in general position with respect to the given family of subspaces. In Theorem~\ref{generic} we will prove that indeed ``general position'' is a generic property, namely holds for almost all hyperplanes. This will complete the connection with the PIT problem solved in this paper.



A {\it hyperplane} in $\K^d$ is a subspace (subspace of $\K^d$) of codimension 1.
Let $F$ be a family of (nonzero) subspaces in $\K^d$ and let $h\subset\K^d$ be a hyperplane in $\K^d$. 
We denote by $F\cap h$ the family $\{f\cap h\mid f\in F\}$.
Following Lov\'asz, we have the following definition:
\begin{defin}[General Position]\label{def:genp}
We say that $h$ is in {\it general position} with respect to $F$ if,
for every $A,B,C \subset F$, with $A$ nonempty, we have:\\
(i) If $\sp(A)\subset h$, then $\sp(A)=\{0\}$.\\
(ii) If\footnote{Note that here one can take any of $A,B,C$ to be the empty set, and we interpret $\sp(\emptyset)=\{0\}$.}
$$
\sp\left((A\cap h)\cup B\right)\cap \sp\left((A\cap h)\cup C\right)\subset h,
$$
then$$
\sp\left((A\cap h)\cup B\right)\cap \sp\left((A\cap h)\cup C\right)= \sp (A\cap h).
$$ 
\end{defin}

\noindent {\it Remark.} In Section~\ref{sec:reduction}, we prove (in Theorem~\ref{generic}) that being in general position with respect to a given family $F$ is a generic property; this fact is mentioned in \cite{Lov} without a proof.

\begin{thm}[{\bf Lov\'asz~\cite[Theorem 2.3]{Lov}}]\label{Lovthm}
Let $F$ be a family of subspaces in $\K^d$. Let $h$ be a hyperplane in $\K^d$ in general position with respect to $F$. Then
$$
\rho_1(F)=d(F\cap h) 
$$
\end{thm}

For completeness, we introduce a slightly more detailed proof, based on the 
line of argument from~\cite{Lov}.

\begin{proof}[Proof of Theorem~\ref{Lovthm}]
Fix $F$ and $h$ as in the statement.
Let $F':=F\cap h$. 
We need to show that $\rho_1(F)=d(F')$.

We first prove that $d(F')\le \rho_1(F)$.
That is, equivalently, we show that 
$
d(F')\le \rho_1(F,\Pi),
$
for every partition $\Pi$ of the family $F$.
Let $\Pi$ be a partition of $F$.
For $P\in\Pi$, let $P':=P\cap h$. Then
$$
\sp({ F'})=\sp\left(\bigcup_{P\in \Pi}\sp(P')\right)
$$
and hence 
$$
d(F')\le\sum_{P\in\Pi}d(P').
$$
Note also that, for every $P\in\Pi$, we have $\sp(P')\subset\sp(P)\cap h$ and hence
$$
d(P')\le d(\sp(P)\cap h)= d(P)-1,
$$
where here we used property (i) of the general position assumption on $h$, namely, we used the fact that $\sp(P)$ is not contained in $h$. We conclude that 
\begin{equation}\label{oneside}
d(F')\le\sum_{P\in\Pi}(d(P)-1),
\end{equation}
for every partition $\Pi$ of $F$. This implies $d(F')\le \rho_1(F)$.

To prove the reverse inequality, we show that, for a certain partition $\Pi^*$ of $F$, the inequality \eqref{oneside} is in fact tight.
We will construct $\Pi^*$ explicitly subsequently refining a given partition. We describe the first step, which is indeed the general step (the proof will allow us to proceed recursively).

Define an equivalence relation on $F$ as follows: For $f_1,f_2\in F$, $f_1\sim f_2$
if and only if 
$$\sp(F'\cup\{f_1\})=\sp(F'\cup\{f_2\}).$$
Let $\{P_1,\ldots,P_m\}$ be the partition (equivalence classes) of $F$ induced by the relation $\sim$.

The main idea is to prove that after intersection with $h$, the spans of the parts $P'_i$ become a direct sum decomposition of $\sp(F')$.
As we will see below, $\Pi^*$ will be achieved by  refining the partition $\{P_1,\ldots,P_m\}$ inductively. 
\begin{lem}\label{lem:dirsum1}
We have \begin{equation}\label{dirsum1}
\sp (F')=\oplus_{i=1}^m \sp(P_i').
\end{equation}
\end{lem}
Before we prove Lemma~\ref{lem:dirsum1}, we establish some preliminary claims.
Let $g_1,\ldots,g_m$ be 
the (distinct) subspaces $g_i:=\sp(F'\cup \{f\})$ for some $f\in P_i$ (note that by construction $g_i$ is independent of the specific element $f\in P_i$ that we take).

We observe that, for every $1\le i\le m$, 
\begin{equation}\label{dimAi}
d(g_i)=d(F')+1.
\end{equation}
Indeed,
by property (i) of general position,  $f$ is not contained in $h$ and $\dim (f\cap h)=\dim(f)-1$, for every $f\in F$. Hence, for every $f\in F$, 
one can choose a basis for $f$ with all elements of the basis in $h$ except for exactly one element $b_f$ which is not in $h$. 
Thus, fixing any $f\in P_i$, we have
$$
g_i=\sp(F'\cup\{f\})=\sp(F'\cup\{b_f\})=\sp(F')\oplus \sp\{b_f\}.
$$
Thus, $d(g_i)=d(F')+1$, as needed.

Next, we observe that, for $i\neq j$, we have 
\begin{equation}\label{Aij}
g_i\cap g_j=\sp(F')\subset h.
\end{equation}
Indeed, by construction $g_i\neq g_j$, and in particular $g_i\cap g_j\subsetneq g_i$.
Combining this with \eqref{dimAi}, we get
$d(g_i\cap g_j)\le d(g_i) -1= d(F')$.
By the definition of $g_i,g_j$,  we also have $\sp(F')\subset  g_i\cap g_j$. Hence $g_i\cap g_j=\sp(F')$ and \eqref{Aij} follows.

\begin{proof}[Proof of Lemma~\ref{lem:dirsum1}]
Here property (ii) of the general position definition will be crucial for the induction step.
If $m=1$ then \eqref{dirsum1} clearly holds. 
For $m\ge 2$, it suffices to show that, for every $2\le k\le m$ and every distinct indices $1\le i_1,\dots,i_k\le m$, one has
\begin{equation}\label{dirsum}
\sp(P_{i_1}'\cup\cdots\cup P_{i_{k-1}}')\cap\sp(P_{i_{k}}')=\{0\}.
\end{equation}

We prove \eqref{dirsum} by induction on $k$.
For $k=2$,
we need to show that $\sp(P_{i_1}')\cap\sp(P_{i_2}')=\{0\}$, for every distinct $1\le i_1,i_2\le m$.
By the definition of the subspaces $g_{i_1},g_{i_2}$ and applying \eqref{Aij}, we have
\begin{align*}
\sp(P_{i_1})\cap\sp(P_{i_2})
&\subset g_{i_1}\cap g_{i_2}\subset h.
\end{align*}
Since $h$ is in general position, using property (ii),  this implies that $\sp(P_{i_1})\cap\sp(P_{i_2})=\{0\}$.
This proves the induction base case $k=2$.

Assume next that \eqref{dirsum} holds for some $2\le k\le m-1$ fixed and for every distinct indices $1\le i_1,\ldots,i_{k}\le m$.
Let $1\le i_1,\ldots,i_{k+1}\le m$ be some distinct indices. To establish the induction step we need to prove
\begin{equation}\label{step}
\sp(P_{i_1}'\cup\cdots\cup P_{i_{k}}')\cap\sp(P_{i_{k+1}}')=\{0\}.
\end{equation}
Observe that in order to prove \eqref{step} it suffices to show that
\begin{equation}\label{dirs}
\sp(P_{i_1}'\cup\cdots\cup P_{i_k}')\cap\sp(P_{i_2}'\cup\cdots P_{i_{k+1}}')\subset\sp(P_{i_2}'\cup\cdots\cup P_{i_k}').
\end{equation}
Indeed, assume that \eqref{dirs} holds. Then
\begin{align*}
\sp(P_{i_1}'\cup\cdots\cup P_{i_k}')\cap\sp(P_{i_{k+1}}')
&=
\sp(P_{i_1}'\cup\cdots\cup P_{i_k}')\cap\sp(P_{i_2}'\cup\cdots\cup P_{i_{k+1}}')\cap\sp(P_{i_{k+1}}')\\
&\subset\sp(P_{i_2}'\cup\cdots\cup P_{i_k}')\cap \sp(P_{i_{k+1}}'),
\end{align*} 
where the first line uses the trivial fact that $\sp(P_{i_{k+1}}')\subset \sp(P_{i_2}'\cup\cdots\cup P_{i_{k+1}}')$ and the
second line is due to \eqref{dirs}.
By the induction hypothesis, we have 
$$
\sp(P_{i_2}'\cup\cdots\cup P_{i_k}')\cap \sp(P_{i_{k+1}}')=\{0\}.
$$
Thus, assuming that \eqref{dirs} is true, \eqref{step} follows.

Finally, we now prove \eqref{dirs}.
Note that, by the definition of the subspaces $g_i$ and using \eqref{Aij}, we have
\begin{align*}
\sp(P_{i_1}\cup(P_{i_2}'\cup\cdots\cup P_{i_k}'))\cap\sp((P_{i_2}'\cup\cdots\cup P_{i_k}')\cup P_{i_{k+1}})\subset g_{i_1}\cap g_{i_{k+1}}\subset h.
\end{align*}
Hence, our assumption that $h$ is in general position with respect to $F$ implies that in fact
\begin{align*}
\sp(P_{i_1}\cup(P_{i_2}'\cup\cdots\cup P_{i_k}'))\cap\sp((P_{i_2}'\cup\cdots\cup P_{i_k}')\cup P_{i_{k+1}})\subset \sp(P_{i_2}'\cup\cdots\cup P_{i_k}').
\end{align*}
This clearly implies \eqref{dirs}. 
Thus we have established the inductive step and 
this completes the proof of Lemma~\ref{lem:dirsum1}.
\end{proof}

Recall that our goal is to show that \eqref{oneside} is tight for some partition $\Pi^*$ of $F$.
In view of Lemma~\ref{lem:dirsum1}, for the partition $\{P_1,\ldots,P_m\}$ defined above, one has
\begin{equation}\label{sumGi'}
d(F')=\sum_{i=1}^md(P_i').
\end{equation}
That is, we expressed the quantity $d(F')$ as the sum of the quantities $d(P_i')$ for certain subfamilies $P_1,\ldots,P_m$ of $F$. 
This allows to prove the existence of $\Pi^*$ using induction on the size of $F$.

If $|F|=1$, the unique partition on $F$ clearly attains \eqref{oneside}.
For $|F|\ge 1$, let $\{P_1,\ldots,P_m\}$ be the partition of $F$ given by Lemma~\ref{lem:dirsum1}, satisfying \eqref{sumGi'}.
If $m=1$, 
the identity \eqref{sumGi'}, combined with \eqref{dimAi}, gives
\begin{equation*}
d(F')=d(P_1)-1.
\end{equation*}
This means that \eqref{oneside} is tight, and thus $\Pi^*=\{P_1\}$.
If $m>1$, then each subfamily $P_i$ has fewer elements than $F$.
Applying the induction hypothesis, there exist subpartitions 
$\Pi_i^*=\{P_{i1},\ldots,P_{im_i}\}$ of $P_i$, for each $1\le i\le m$, satisfying
$$
d(P_i)=\sum_{j=1}^{m_i}(d(P_{ij})-1).
$$
Combined with \eqref{sumGi'}, we get
$$
d(F')=\sum_{i=1}^m\sum_{j=1}^{m_i}(d(P_{ij})-1).
$$
So $\Pi^*:=\bigcup_{i=1}^m\Pi_i^*$ forms a partition of $F$ that attains \eqref{oneside}. 
This completes the proof of the theorem.
\end{proof}

\begin{remark}\label{remarkLovthm}
Note that in the inductive proof of Lemma~\ref{lem:dirsum1}, it was sufficient to consider not {\em all} $k$-subsets of the $P_i$ in the given partition, but rather simply on intervals $P_2,P_3,\dots ,P_k$. The same induction on $k$ works without change. Thus even after refinement, in the proof of this theorem we never need to apply the 
 ``general position'' condition more than $|F|$ times. This will help us later bound  the show that $\rho_1(F)$ correctly computes $\dim (F \cap h)$ for most (or generic) hyperplanes $h$ even when $\K$ is finite and not too large.
 \end{remark}

We now generalize the theorem above to intersecting a family of subspaces with an arbitrary subspace. For this we need to extend the definition of ``general position''. 

Let $F$ be a family of subspaces in $\K^d$. Let $\{\x_1,\ldots,\x_k\}$ be a set of vectors, and define that the subspaces  $h_i=\{\x_1,\ldots,\x_i\}^\perp$.  Note that $h_i$ is of codimension $i$ in $\K^d$, and that $h'_i:=h_i\cap h_{i-1}$ is a hyperplane in $h_{i-1}$, for $i=1,\ldots,k$. We say that the subspace $h=h_k$ is in {\em general position} with respect to $F$ if for all $i\in [k]$ we have that the hyperplane $h'_i$ is in general position with respect to the family $F_i = F\cap h_{i-1}$.

\begin{thm}\label{codimk}
Let $F$ be a family of subspaces in $\K^d$. Let $h$ be a subspace in $\K^d$ of codimension $k$ in general position with respect to $F$. Then
$$
\rho_k(F)=d(F\cap h) 
$$
\end{thm}

\begin{proof}
We prove by induction on the codimension $k$. The case $k=1$ is Theorem~\ref{Lovthm}.

Let $\x_1,\ldots,\x_k\in\K^d$ be vectors such that $h=\{\x_1,\ldots,\x_k\}^\perp$ is in general position with respect to $h$. 
We know that  $h'_k$ is in general position with respect to the family $F_k:=F\cap h_{k-1}$.
By Theorem~\ref{Lovthm} again, we have 
\begin{align*}
d(F\cap h)=d(F_k\cap h'_{k})&=\rho_1(F_k)\\
&=\min_{\Pi_k}\sum_{P'\in\Pi_k}(d(P')-1),
\end{align*} 
where the minimum ranges over all partitions $\Pi_k$ of $F_k$.
Note that $\Pi_k$ induces a partition $\Pi$ on $F$, in the obvious way.
Moreover, for every $P'\in \Pi_k$ there exists $P\subset F$ such that $P'=P\cap h_{k-1}$. By induction, we get
$$
d(P')=d(P\cap h_{k-1})=\rho_{k-1}(P).
$$
Thus,
\begin{align*}
d(F\cap h)
&=\min_\Pi\sum_{P\in \Pi}(\rho_{k-1}(P)-1)\\
&=\min_\Pi\sum_{P\in \Pi}\left(\left(\min_{\Pi_P}\sum_{Q\in\Pi_P}(d(Q)-k+1)\right)-1\right),
\end{align*}
where the first minimum (the outer one) in this exprssion is taken over all partitions $\Pi$ of $F$, and, fixing $\Pi$ and given $P\in \Pi$, the inner minimum is taken over all partitions $\Pi_P$ of the family $P$.

Note that, for any partition $\Pi$ of $F$, the partitions $\{\Pi_P\mid P\in \Pi\}$ induce a new partition $\Pi'$ which is a refinement of $\Pi$. Namely, $\Pi':=\bigcup_{P\in \Pi}\Pi_P$.
Note that taking $\Pi_P=\{P\}$ for each $P\in \Pi$, we get 
\begin{align}
d(F\cap h)
&\le \min_{\Pi}\sum_{P\in \Pi}\left(\left(\sum_{Q\in\{P\}}(d(Q)-k+1)\right)-1\right)
\nonumber\\ 
&=\min_{\Pi}\sum_{P\in \Pi}(d(P)-k)\nonumber\\
&=\rho_k(F).\label{lerhok}
\end{align}

We now prove the inverse inequality. Fix a partition $\Pi$ of $F$, and, for $P\in \Pi$, let $\Pi_P^*$ be a partition of $P$ that attains the minimum in 
$$
\min_{\Pi_P}\sum_{Q\in\Pi_P}(d(Q)-k+1).
$$
That is, the partitions $\Pi_P^*$ satisfy
$$
\sum_{P\in \Pi}\left(\left(\min_{\Pi_P}\sum_{Q\in\Pi_P}(d(Q)-k+1)\right)-1\right)=
\sum_{P\in \Pi}\left(\left(\sum_{Q\in\Pi_P^*}(d(Q)-k+1)\right)-1\right)
$$
Let $(\Pi')^*$ be the partition of $F$ induced by $\bigcup\{\Pi_P^*\mid P\in \Pi\}$.
Observe that
\begin{align}
d(F\cap h)
&=\min_\Pi\sum_{P\in \Pi}\left(\left(\sum_{Q\in\Pi_P^*}(d(Q)-k+1)\right)-1\right)
\nonumber\\
&\ge\min_\Pi\sum_{P\in \Pi}\sum_{Q\in\Pi_P^*}((d(Q)-k+1)-1)
\nonumber\\
&=\min_\Pi\sum_{Q\in (\Pi')^* }(d(Q)-k)
\nonumber\\
&=\min_{(\Pi')^*}\sum_{Q\in (\Pi')^* }(d(Q)-k)\nonumber\\
&\ge \min_{\Pi}\sum_{Q\in \Pi }(d(Q)-k)\nonumber\\
&=\rho_k(F).\label{gerhok}
\end{align}
Combining the inequalities \eqref{lerhok} and \eqref{gerhok}, we get $d(F\cap h)=\rho_k(F)$. This completes the induction step, and therefore proves the theorem.
\end{proof}


\section{Rank of symbolic matrices}\label{sec:reduction}

In this section we
show that  the quantity $\rho_c(F)$ can be interpreted as the generic rank, defined as the rank over $\K(\x)$, of a certain symbolic matrix associated with $F$.
More concretely, for $\x\in\K^d$ let  
$$
h(\x):=(\sp\{\x\})^\perp.
$$ We prove that $\rho_c(F)$ equals to the generic rank of a symbolic matrix whose  entries are linear combinations of the coordinates of $\x$.
 
Our main result for the section is the following (note that this is Theorem~\ref{R2} in the introduction).
\begin{thm}\label{PIT}
Let $u_1,\ldots,u_n,v_1,\ldots,v_n\in\K^d$ be row vectors. Consider the symbolic matrix $A(\x)$, with unknowns $\x=(x_1,\ldots,x_d)$, whose $i$th row is 
$$
(v_i^t u_i-u_i^t v_i)\x$$
Then the (generic) rank of $A(\x)$ can be computed in polynomial time.
 \end{thm}

To prove the theorem we use the property established in Theorem~\ref{Lovthm}, interpreting the quantity $\rho_1(F)$ as the dimension of the space spanned by 
$$
F\cap h=\{f\cap h\mid f\in F\},
$$ 
for any hyperplane $h$ in general position with respect to $F$ (see Definition~\ref{def:genp}).
Taking $h=h(\x)$ we prove, in Lemma~\ref{fcaph}, that the intersection $f\cap h(\x)$ is the span of vectors with entries that are linear combinations of the coordinates of $\x$.
We then prove, in Theorem~\ref{generic},  that, given a family $F$, $h(\x)$ is in general position with respect to $F$, for every generic $\x$ (namely, for almost every $\x \in \K^d$). Finally, we use the algorithm for computing $\rho_1$ from Section~\ref{sec:algorithm}.

\begin{lem}\label{fcaph}
Let $f$ be an $m$-dimensional subspace in $\K^d$ and let $v_1,\ldots,v_m$ be a basis of $f$.
Let $\x\in\K^d$ and assume that $f\not\subseteq h(\x)$.
Then $h(\x)\cap f$ is spanned by vectors of the form
$$
w_{ij}:=(v_j\cdot \x)v_i-(v_i\cdot \x)v_j,
$$
with $i\neq j$. \\
Moreover, if (wlog) $\x\cdot v_1\neq 0$, then the set
$\{w_{12},\ldots,w_{1m}\}$
forms a basis of $f\cap h_{\x}$.
\end{lem}
\begin{proof}
We first observe that $w_{ij}\in f\cap h(\x)$.
Indeed, by definition, each $w_{ij}$ is a linear combination of basis vectors for  $f$, and thus $w_{ij}\in f$.
We also have
\begin{align*}
w_{ij}\cdot \x
&=((v_j\cdot \x)v_i-(v_i\cdot \x)v_j)\cdot \x\\
&=(v_j\cdot \x)(v_i\cdot \x)-(v_i\cdot \x)(v_j\cdot \x)=0.
\end{align*}
Thus $w_{ij}\in f\cap h(\x)$.

We now show that $w_{ij}$ also span $f\cap h(\x)$. Indeed, we prove the stronger ``moreover'' statement.

Let $w\in f\cap h(\x)$. Since $w\in f$ we may write $w=\sum_{i=1}^ma_iv_i$.
Since $w\in h(\x)$, we have $w\cdot \x=0$ or 
\begin{equation}\label{w}
0=\sum_{i=1}^ma_iv_i\cdot \x.
\end{equation}
If $v_i\cdot \x=0$ for every $i$, then $f\subseteq h(\x)$, contradicting our assumption. We may therefore assume, without loss of generality, that $v_1\cdot \x\neq 0$. In this case \eqref{w} can be rewritten as
$$
a_1=-\sum_{i=2}^m\frac{a_iv_i\cdot \x}{v_1\cdot \x}.
$$
We conclude that
\begin{align*}
w
&=\sum_{i=1}^ma_iv_i\\
&=-\left(\sum_{i=2}^m\frac{a_iv_i\cdot \x}{v_1\cdot \x}\right)v_1+\sum_{i=2}^ma_iv_i\\
&=\sum_{i=2}^m\frac{-a_i}{v_1\cdot \x}\left((v_i\cdot \x) v_1-(v_1\cdot \x)v_i\right)\\
&=\sum_{i=2}^m\frac{-a_i}{v_1\cdot \x}w_{1i}.
\end{align*}
 This completes the proof of the lemma.
\end{proof}

We observe an interesting consequence of Lemma~\ref{fcaph}, 
asserting that computing $\rho_1(F)$ for a family $F$ can be reduced to computing $\rho_1(G)$, for a certain family $G$ consisting only of planes (two-dimensional subspaces). 
\begin{cor}\label{remark}
Let $F=\{f_1,\ldots,f_n\}$ be a family of subspaces in $\K^d$ and let
$\{v_{i1},\ldots,v_{im_i}\}$ be a basis of $f_i$, for $i=1,\ldots,n$.
Consider the family of two-dimensional subspaces
$$
G=\bigcup_{i=1}^n\{g_{ijk}\mid 1\le j\neq k\le m_i\},
$$
where 
$
g_{ijk}=\sp\{v_{ij},v_{ik}\}.
$
Then $\rho_1(F)=\rho_1(G)$.
\end{cor}
\begin{proof}
It follows easily from Theorem~\ref{generic} that $h(\x)$ is in general position with respect to both families $F$ and $G$, for every generic $\x\in\K^d$.
Fixing such $\x\in\K^d$ and applying Lemma~\ref{fcaph}, we see that 
$\sp(F\cap h(\x))=\sp(G\cap h(\x))$. By Theorem~\ref{Lovthm} this means that $\rho_1(F)=\rho_1(G)$, as needed.
\end{proof}

The following lemma is a natural extension of Lemma~\ref{fcaph} to a similar description of the intersection of a given subspace with a generic one, where the latter is not necessarily of co-dimension 1. If the co-dimension is $k$,  the basis elements of the intersection will be homogeneous polynomials of degree $k$  in the entries of the generic vectors. This connection, together with our algorithm for computing $\rho_k$, will prove Theorem~\ref{Rk} from the introduction.
\begin{lem}\label{fcaphk}
Let $k<m\le d$ be integers.
Let $f$ be an $m$-dimensional subspace in $\K^d$ and let $v_1,\ldots,v_m$ be a basis of $f$.
Let $\x_1,\ldots,\x_k$ be vectors  in $\K^d$ and define the subspace
$$h:=\left(\sp\{\x_1,\ldots,\x_k\}\right)^\perp.$$
Assume that $\dim(f \cap h)=m-k$ (this extends the assumption $f\not\subseteq h(\x)$ of the lemma above).
Let $X$ be the $k\times d$ matrix with $\x_i$ as its $i$th row. 
Let $V$ denote the $d\times m$ matrix with $v_j$ as its $j$th column.
Put $M:=XV$. So $M$ is a $k\times m$ matrix with $(i,j)$ entry being $\x_i\cdot v_j$.
For every $I\subset [m]$ of cardinality $k$, let $M_I$ denote the $k\times k$ matrix received by restricting to the columns of $M$ with indices in $I$.
Then $f\cap h$ is the span of vectors of the form
$$
w_S:=\sum_{j=1}^{k+1}(-1)^j\det(M_{I_j})v_{s_j},
$$
where $S=\{s_1<\ldots<s_{k+1}\}\subset [m]$ is of cardinality $k+1$ and $I_j:=S\setminus \{s_j\}$.\\

Moreover, if (wlog, given our assumption above), assuming that the last $k$ columns of M are linearly independent,  
 $f\cap h$ is spanned by the $m-k$ vectors $w_S$ with $S$ containing the last $k$ columns.

\end{lem}

\begin{proof}
We first show that $w_S\in f\cap h$, for every $S\subset [m]$ of cardinality $k+1$.
For $S$ fixed, we need to verify that $w_S$ is orthogonal to each of $\x_1,\ldots,\x_k$.
For every $1\le i\le k$ we have 
$$
w_S\cdot \x_i=\sum_{j=1}^{k+1}(-1)^j\det(M_{I_j})v_{s_j}\cdot \x_i.
$$
Observe that the right-hand side is exactly the determinant of the matrix received by duplicating the $i$th row of $M$. Since the latter matrix is evidently singular, we conclude that $w_S\cdot \x_i=0$, for every $i=1,\ldots,k$.
Thus $w_S\in h$. Clearly, we also have $w_S\in f$. Thus $w_S\in f\cap h$, as needed.

We now turn to prove that the vectors $w_S$ generate $f\cap h$. Indeed we prove the stronger ``moreover'' statement that already the $m-k$ vectors $w_S$ with $S$ of size $k+1$ that  contain the last $k$ columns span $f\cap h$. Recall that the last $k$ columns of $M$ are independent. 

It will be convenient to add one more piece of (slightly informal) notation. Let $M'$ be the matrix extending $M$ with one more (say, 0'th) row, that contains in the $j$th coordinate the {\it vector} $v_j$. Note that, up to a sign, the determinant of any $k+1$ minor of $M'$ on columns $S$ is precisely $w_S$.

Note also that column operations on $M'$, and replacing $w_S$ by the $k+1$ minors of the resulting matrix, do not change the span of the vectors $w_S$.
Moreover, note that column operations on the last $k$ columns of $M'$ do not change the vectors $w_S$, restricting to sets $S\subset I$ of size $k+1$ that contain the indices of the last $k$ columns.
We may therefore assume, by performing such column operations, that the last $k$ columns of $M$ form the $k\times k$ identity matrix.


We will prove the lemma by induction on $k$. We already know that this statement holds for $k=1$ (and any $m$) by Lemma~\ref{fcaph}. Assume it holds for $k-1$ (and $m-1$, this is all we need), 
and we will infer the statement for $k$. Consider  the subspace $h'$ orthogonal to the vectors $\x_1, \dots ,\x_{k-1}$, and the subspace $f'$ spanned by the vectors $v_1,\dots ,v_{m-1}$, and form the associated $(k-1)\times (m-1)$ matrix, say $N$. Add to the matrix $N$ the $0'th$ row to create $N'$.
By induction, we know that the $k$-minors containing the last $k-1$ columns of $N'$ are vectors which span the $f' \cap h'$. For $i\in[m-k]$, let $w_i'$ denote the basis vector that corresponds to the columns $\{i,m-k+1,\ldots,m-1\}$.
Note that $$f\cap h=\sp((f'\cap h')\cup \{v_m\})\cap \{\x_k\}^\perp.$$

Now add to $N'$ a last column for $v_m$ and a last row for $x_k$ to form $M'$. Fix $i\in[m-k]$, and write $w_i:=w_{S_i}$, where $S_i=\{i,m-k+1,\ldots,m\}$. 
Due to the last $k$ columns of $M$ being the identity matrix, we have
$$
w_i=(\x_k\cdot v_i) v_m-w_i'.
$$
Moreover, one can check that in fact
\begin{align*}
\x_k\cdot v_i &=\x_k\cdot w_i' ~~~~~~\text{and}\\
w_i'&=(\x_k\cdot v_m) w_i'.
\end{align*}
That is, $w_i=(\x_k\cdot w_i') v_m-(\x_k\cdot v_m) w_i'$.
Applying Lemma~\ref{fcaph}, we get that the vectors $w_i$, for $i\in[m-k]$, form a basis for $f\cap h$, as needed.
\end{proof}

\section{Generic vs. General Position}\label{sec:generic}


This section completes the cycle of connections, proving that most (namely, generic) hyperplanes, and indeed most subspaces, are in general position (in the Lov\'asz sense of Section~\ref{sec:Lov}) with respect to any given family of subspaces. The proof will make use the explicit description we established in the previous section for a basis to the intersection of a family of subspaces and a hyperplane. Thus, computing the ranks of the symbolic matrices in Theorems~\ref{R2} and \ref{Rk} are equivalent to computing the functions $\rho_1$ and $\rho_k$ respectively, which we can do efficiently by the algorithm of Section~\ref{sec:alg}. 


\begin{thm}\label{generic}
Let $F$ be a family of subspaces in $\K^d$, and assume that either $\char(\K) > |F|$ or $\char(\K)=0$. 
Then the hyperplane  $h(\x)$ is in general position (see Definition~\ref{def:genp}) with respect to $F$ for almost every $\x\in\K^{d}$. More precisely, over finite fields all but $|F|/|\K|$- fraction of hyperplanes are not in general position, and for infinite fields they have measure zero.
\end{thm}

The proof of this theorem turns out to be more intricate than we imagined. 
We will give below a linear-algebraic proof that is valid for all fields $\K$. 
In the appendix we give an alternative, geometric proof which is valid for the field $\R$ of Real numbers.

\begin{proof} 
Fix subsets $A,B,C\subset F$. Our goal is to show that for
$$
S:=\sp_\K((A\cap h(\x))\cup B)\cap \sp_\K((A\cap h(\x))\cup C)
$$
either $S\not\subseteq h(\x)$ generically, or $S\subset A\cap h(\x)$ generically. Indeed, we will prove that one of these alternative holds for every $\x$, except for those $\x$ that vanish on a certain nontrivial linear equation.
Thus, if $\K$ is finite, the fraction of such exceptional values of $\x$ is $1/|\K|$.
Since the number of choices of $A,B,C$ is finite, we see that if $\K$ is large enough this probability remains negligible. Being a bit more careful, (see Remark~\ref{remarkLovthm} at the end of the proof of Theorem~\ref{Lovthm}), there are at most $|F|$ applications  of the ``general position'' definition, and so the fraction of ``bad'' $\x$ is at most $|F|/|\K|$ as stated.

It is easy to see that replacing $B$ by $\sp B$ and $C$ by $\sp C$ does not affect the subspace $S$. We may therefore assume that each of the families $B,C$ contains a single subspace of $\K^d$. 

Suppose that $B\cap C\neq \{0\}$, that is, that there exists $v\in B\cap C$, with $v\neq 0$. Clearly, we have $v\in S$ and the linear form $v\cdot\x$ not identically zero. Thus, for almost every $\x$, $S$ is not contained in $h(\x)$ 
and there is nothing to prove in this case. We may therefore assume that $B\cap C=\{0\}$. In this case, after a change of basis of $\K^d$, we may assume that $B=\sp\{e_1, \ldots,e_k\}$ and $C=\{e_{k+1},\ldots,e_{k+m}\}$, where $1\le k<k+m\le d$ and  $e_1,\ldots,e_d$ stand for the standard basis vectors in $\K^d$.

From now on we will regard $\x$ as a vector of variables, and work in the field of fractions $\K(\x)$. In particular this makes all subspaces under consideration, $A,B,C$, $A\cap h(\x)$ and of course $S=S(\x)$ now subspaces of $\K(\x)^d$
(by taking the span of their bases in $\K(\x)^d$).  

With this, our task becomes proving the following about these subspaces: 

\begin{clm}\label{rational}
Either $S\not\subseteq h(\x)$, or $S\subset A\cap h(\x)$.
\end{clm}

We will break this task to two. Clearly, it will suffice to prove the claim for any spanning set $S'$ replacing $S$.
So first we will prove that we can take $S'$ to be the affine functions (of $\x$) in $S$, and then we will prove the claim for $S'$.

\begin{lem}\label{lemS0}
$S$ is spanned by its elements which are affine functions of $\x$.
\end{lem} 

\begin{proof}[Proof of Lemma~\ref{lemS0}]
Recall that we showed, in Lemma~\ref{fcaph}, that $\sp_\K(A\cap h)$ has a basis consisting of elements of the form  $(u^tv-v^t u)\x$, for some $u,v\in\K^d$. Write $\{\a_1(\x),\ldots,\a_n(\x)\}$ for a basis of $\sp_\K(A\cap h)$ of this form.

Having bases for $B,C$ and $A\cap h(\x)$ we can express all elements of $S$ as linear combinations of these bases. Thus, elements in $S$ are described by solutions $\alpha, \alpha' \in \K^n$,
$\beta \in \K^k$, $\gamma \in \K^m$  to the following 
system of linear equations. 
\begin{equation}\label{sysinter}
\sum_{i=1}^n \alpha_i \a_i(\x)+\sum_{i=1}^k\beta_i e_i=\sum_{i=1}^n \alpha'_i \a_i(\x)+\sum_{i=1}^{m} \gamma_i e_{k+i}
\end{equation}
where $\alpha_i\in\K $ (resp., $\alpha_i', \beta_i,\gamma_i\in \K$) is the $i$th entry of $\alpha$ (resp., $\alpha',\beta,\gamma$).

By basic theory of linear algebra, there exists a set of solutions, each of the form 
\begin{equation}\label{intersysx}
w=w(\x)=\sum_{i=1}^n \alpha_{i}(\x) \a_i(\x)+\sum_{i=1}^k\beta_{i}(\x)e_i=\sum_{i=1}^n \alpha'_{i}(\x) \a_i(\x)+\sum_{i=1}^m \gamma_{i}(\x) e_{k+i},
\end{equation}
where $\alpha_{i}(\x),\alpha'_{i}(\x), \beta_{i}(\x),\gamma_{i}(\x)$ are rational functions in the entries of $\x$, that together span the subspace $S$.
Moreover, these rational functions are of degree at most $|F|$.

We will now strive to find a simpler spanning set $S'$ for $S$, and then use it to prove Claim~\ref{rational}.

The first simplification is realizing (via common denominators) that  without loss of generality we can assume that all $\alpha_{i}(\x),\alpha'_{i}(\x), \beta_{i}(\x),\gamma_{i}(\x)$ are in fact {\em polynomials} in the entries of $\x$. These elements of $S$  span the rest, after dividing by some fixed polynomial.

The next simplification (separating out homogeneous terms) shows that without loss of generality we can take all the polynomials in each of $\alpha, \alpha', \beta, \gamma$ to be homogeneous of the same degree, which we may respectively  call $\deg(\alpha), \deg(\alpha'), \deg(\beta), \deg(\gamma)$. These homogeneous solutions certainly span $S$, and now we refine their structure further.

Indeed, inspecting the system of equations we know more: since each entry of $\a_i(\x)$, for every $i$ is of degree one, we know that for some fixed integer $r\geq 0$, they must satisfy $\deg(\alpha)= \deg(\alpha') =r$ and $ \deg(\beta) = \deg(\gamma) = r+1$. We   use this to stratify solutions $w$ by degree, and say that the associated $w$ has degree $r$.  Let $S_r$ be all solutions of degree $r$ (note that each $S_r$ is a subspace over $\K$, though we will not use this fact). We call solutions $w$ of degree 0 {\em linear}. Our main simplification will come from showing that linear elements $S_0$ span $S$, which in this notation is a restatement of the lemma we are proving.

\begin{clm}
$\sp S_0 = S$
\end{clm}
 
We will prove this claim by induction on $r$, using our stratifications $S_r$ of members of $S$. It is clearly true for $r=0$. So assume $S_0$ spans $S_r$, and we need to prove that $S_0$ spans $S_{r+1}$. By induction, it suffices to prove that $S_r$ spans $S_{r+1}$. The plan for this will be as follows. We will assume we have some $w\in S_{r+1}$. We will take all partial derivatives of its constituent polynomials with respect to each variable $x_t$, $t\in [d]$. From each of these we will generate an element $w_t \in  S_r$, as the degree decreased by 1. Finally, we will show that $w$ is a linear combination, indeed a very simple one, of the form : $(r+1)w= \sum_{t=1}^d x_t w_t$. We now elaborate.

Fix $t\in [d]$.
Let us take a  derivative with respect to the variable $x_t$ of $\x$, of both sides of the identity \eqref{intersysx}. We get
\begin{equation*}
\sum_{i=1}^n \left(\frac{\partial \alpha_{i}(\x)}{\partial x_t} \a_i(\x)+
\alpha_{i}(\x)\frac{\partial \a_i(\x)}{\partial x_t}\right)
+\sum_{i=1}^k\frac{\partial\beta_{i}(\x)}{\partial x_t}e_i =
\end{equation*}
\begin{equation*}
\sum_{i=1}^n \left(\frac{\partial \alpha_{i}'(\x)}{\partial x_t} \a_i(\x)
+\alpha'_{i}(\x)\frac{\partial \a_i(\x)}{\partial x_t}\right)
+\sum_{i=1}^m \frac{\partial\gamma_{i}(\x)}{\partial x_t} e_{k+i}
\end{equation*}

To define $w_t$ we first define $\alpha(t), \alpha'(t), \beta(t), \gamma(t)$ by appropriately collecting homogeneous terms, and making sure that $\alpha(t), \alpha'(t) \in A\cap h$ are of degree $r$, and that $\beta(t) \in B$ and $\gamma(t) \in C$ are of degree $r+1$:
\begin{itemize}
\item $\alpha(t)_{i}=\frac{\partial \alpha_{i}(\x)}{\partial x_t}$
\item $\alpha'(t)_{i}=\frac{\partial \alpha'_{i}(\x)}{\partial x_t}$,
\item For $i\in [k]$,  $\beta(t)_{i}(\x)$ is 
$$
\left[\sum_{s=1}^n(\alpha_{s}(\x)-\alpha'_{s}(\x))\frac{\partial \a_s(\x)}{\partial x_t}\right]_i+\frac{\partial \beta_{i}(\x)}{\partial x_t}
$$ 
\item For $i\in [m]$, $\gamma(t)_{i}(\x)$ is 
$$
\left[\sum_{s=1}^n(\alpha_{s}'(\x)-\alpha_{s}(\x))\frac{\partial \a_s(\x)}{\partial x_t}\right]_{k+i}+\frac{\partial\gamma_{i}(\x)}{\partial x_t};
$$
\end{itemize}
here we used $[v]_j$ to denote the $j$th entry of a vector $v$.
Now we can formally define $w_t \in S_r$ as follows.
We first observe that
\begin{equation}\label{smallerdeg}
 \sum_{i=1}^n \alpha(t)_{i}(\x)\a_i(\x)
+\sum_{i=1}^k \beta(t)_{i}(\x) e_i \\
= \sum_{i=1}^n \alpha'(t)_{i}(\x) \a_i(\x)
+\sum_{i=1}^m \gamma(t)_{i}(\x) e_{k+i}.
\end{equation}
Indeed, note that \eqref{intersysx}, restricted to the $j$th component of the equation, implies that for every, $k+m<j\le n$, we have
$$
\left[\sum_{i=1}^n (\alpha'_{i}(\x) - \alpha_{i}(\x)) \a_i(\x)\right]_{j}=0.
$$
From this it is straightforward to verify that the identity \eqref{smallerdeg} indeed holds. Thus, letting 
$$
w_t:= \sum_{i=1}^n \alpha(t)_{i}(\x)\a_i(\x)
+\sum_{i=1}^k \beta(t)_{i}(\x) e_i,
$$
for each $t$, the identity \eqref{smallerdeg} implies that $w_t$ is in $S$. Moreover, by our definition, $w_t$ is of degree $r-1$.

It remains to prove that $w$ is spanned by the vectors $w_t$.
For this, one basic fact we will need is that if  $p(\x)$ is any homogeneous polynomial of degree $m$, it satisfies
$$
\sum_t x_t\cdot \frac{\partial p(\x)}{\partial x_t}=m p(\x).
$$
The second fact we will need follows from identity~\eqref{intersysx}, when restricted to the $j$th component of the equation. For every $j\in [k]$,
$$\left[\sum_{i=1}^n (\alpha_{i}'(\x) - \alpha_{i}(\x)) \a_i(\x)\right]_j = \beta_j.$$

Combining these two properties, we get 
\begin{itemize}
\item $\sum_t x_t \alpha(t) = r\alpha $
\item $\sum_t x_t \beta(t) = r\beta $
\end{itemize}
and this implies that
$$
rw=\sum_t x_tw_t.
$$
Note that $r\neq 0$; indeed, for $\K$ with non-zero characteristic, we have $r<\char(\K)$. Thus the vectors $w_t$ span $w$. 
This completes the induction step, and hence the proof of Lemma~\ref{lemS0}.
\end{proof}

To complete the proof of the theorem we now prove
\begin{lem}\label{lemMatrep}
Either $S_0$ is not contained in $h(\x)$, or it is contained in $A \cap h(\x)$.
\end{lem}

As the elements in $S_0$ are affine functions of $\x$, a violation of the first possibility will imply that $\x$ satisfy a linear equation, so the fraction of such vectors is at most $1/|\K|$ as requested.

\begin{proof}[Proof of Lemma~\ref{lemMatrep}]


We first introduce some notation. Let $v(\x)$ be a vector in $\K(\x)^d$, such that each entry of $v(\x)$ is some linear combination of $x_1,\ldots,x_d$, the coordinates of $\x$. Then $v(\x)$ can be represented by a matrix $M\in {\rm Mat}_{d\times d}(\K)$, with constant entries, such that $M\x=v(\x)$. 
Note that if $M$ is skew-symmetric, this means that $(M\x)\cdot \x=(M^t\x)\cdot \x=-(M\x)\cdot\x$ or $2(M\x)\cdot \x=0$, which means that $(M\x)\cdot \x=0$, unless the characteristic of the field is $2$. 
Conversely, if $M\x\cdot \x=0$ for every $\x\in \K^d$ and so
$M\x\cdot \x$ is the zero polynomial (in $d$ variables), which implies that $M$ is skew-symmetric.

Consider $k$ such matrices $M_1,\ldots,M_k$, representing vectors $v_1(\x),\ldots,v_k(\x)$, respectively. Then a linear combination $\sum_{i=1}^k\alpha_iM_i$ is a matrix that corresponds to a vector which is a linear combination of $v_1(\x),\ldots,v_k(\x)$, namely, $v(x)=\sum_i\alpha_iv_i(\x)$. Thus $v(\x)$ lies in the span of the vectors $v_i(\x)$.

%

Assume first that $k+m=d$. 
We regard a $(k+m)\times (k+m)$ matrix $M$ as a block matrix with ${\rm TL}(M)$ (resp., ${\rm TR}(M)$, ${\rm BL}(M)$, ${\rm BR}(M)$) denoting the top-left (resp., top-right, bottom-left, bottom-right) blocks. More precisely,
${\rm TL}(M)$ (resp., ${\rm TR}(M)$, ${\rm BL}(M)$, ${\rm BR}(M)$)
stands for the submatrix induced by taking the first  $k$ (resp., first $k$, last $m$, last $m$)  rows and first $k$ (resp., last $m$, first $k$, last $m$) columns of $M$.

With some abuse of notation, we write $M\in Y$, for a subspace $Y$ of $\K(\x)^d$, if $M\x\in Y$.
Recall that $M$ is in $h$ if and only if $M$ is skew-symmetric. In particular,  $TR(M)=-BL(M)^t$, for every  $M\in \sp(A\cap h)$.
Assume that for some $M\in \sp(A\cap h)$, we have ${\rm TR}(M)\neq 0$ (and thus also $BL(M)\neq 0$). We claim that in this case there exists a matrix $\widetilde M\in S\setminus h$. To see this it is sufficient to show that there exist matrices $b\in B$ and $c\in C$ such that $M+b=c$ which is not skew-symmetric (and therefore not in $h$).
Indeed, let $b$ be defined by $TL(b)=-TL(M)$, $TR(b)=-TR(M)$, and $BL(b)=BR(b)=0$. 
We define the matrix $c$ by $TL(c)=TR(c)=0$, $BL(c)=BL(M)$, $BR(c)=BR(M)$.
Clearly, $b\in B$, $c\in C$ and $M+b=c$. 
If $c$ is skew-symmetric, then we must have $BL(c)=BL(M)=0$, contradicting our assumption on $M$. Thus $c=M+b$ is in $A\cap h$ but not in $S$.
We conclude that in this case the general position requirement holds generically.

Assume next that for every $M\in \sp(A\cap h)$, we have ${\rm TR}(M)={\rm BL}(M) =0$. Recall that $\sp(A\cap h)$ is spanned by matrices of the form $v^tu-u^tv$ for some $u,v\in\K^d$. Assume that $TR(v^tu-u^tv)=BL(v^tu-u^tv)=0$ for such a matrix.
We claim that in this case at least one of $TL(v^tu-u^tv)$ or $BR(v^tu-u^tv)$ is the zero matrix.
Indeed, put $M=v^tu-u^tv$, and assume that $TL(M)\neq 0$.
The for some $1\le i_0\neq j_0\le k$ we have $u_{i_0}v_{j_0}\neq u_{j_0}v_{i_0}$. In particular, not both $u_{i_0}v_{j_0}$ and $u_{j_0}v_{i_0}$ are zero. Assume, without loss of generality, that $u_{i_0}v_{j_0}\neq 0$. 
That is, $u_{i_0},v_{j_0}\neq 0$.
Suppose that $u_\ell=0$ for every $\ell>k$.
In this case it is clear that $BR(M)=0$ and the claim is proved.
Therefore, we may assume that for some $\ell>k$ we have $u_\ell\neq 0$. 
Since we  $BL(M)=0$, we have in particular $u_\ell v_j=u_jv_\ell$, for every $j=1,\ldots,k$. 
In particular, $u_\ell v_{j_0}=u_{j_0}v_\ell$. Note that since $v_{j_0}\neq 0$ and 
$u_\ell\neq 0$, we must have that also $v_\ell, u_{j_0}\neq 0$.
Thus, we get $\frac{v_{i_0}}{u_{i_0}}=\frac{v_\ell}{u_\ell}$ and $\frac{v_{j_0}}{u_{j_0}}=\frac{v_\ell}{u_\ell}$. Combining these equalities, we get that 
$u_{i_0} v_{j_0}=u_{j_0}v_{i_0}$, contradicting our assumption.
This proves the claim.

This implies that $\sp(A\cap h)$ is a direct sum $U\oplus V$ of matrices
with entries supported only on $TL(M)$ for $M\in U$ and matrices supported by $BR(M)$ for $M\in V$.

Now let $w\in S$. By the definition of $S$, $w$ can be written as $w=a+b=a'+c$ for some $a,a'\in \sp(A\cap h)$, $b\in B$, $c\in C$. 
Write $a=a_U+a_V$, where $a_U\in U$ and $a_V\in V$. Similarly, write $a'=a'_U+a'_V$. Then $a_U+a_V+b=a'_U+a'_V+c$, or $a_U-a'_U+b=a'_V-a_V+c$. 
But then, we must have $b=a'_U-a_U$ and $c=a_V-a'_V$, which in particular implies that $b,c\in \sp(A\cap h)$.

Since $a_U-a'_U\in U$ and $a'_V-a_V\in V$, this implies that, without loss of generality, we may assume $a\in U$ and $a'\in V$.
Thus also $w=a+b=a'+c\in \sp(A\cap h)$.
We conclude that $w\in \sp(A\cap h)$ for every $w\in S$. Thus the general position requirement holds in this case.

We now prove the remaining case where $k+m<d$, by reducing it to the case $k+m=d$ just discussed.
Write $k+m=d-z$, for some $z>0$. Repeat the above argument ignoring the last $z$ rows and last $z$ columns of every matrix used along the proof. Note that for $a\in A\cap h$, $a$ is skew-symmetric, and adding a matrix $b\in B$ or $c\in C$ will result with a matrix which is either in $h$ or not in $h$, independent of the last $z$ rows and columns of $a$. Indeed, for $b\in B$ and $c\in C$ these rows and columns are zero, and therefore they cannot affect the skew-symmetry of $a+b$ or $a'+c$. 
\end{proof}
This completes the proof of Theorem~\ref{generic}.
\end{proof}

Having established the connection between genericity and general position, we can now complete the proof of Theorem~\ref{PIT}.

\begin{proof}[Proof of Theorem~\ref{PIT}.]
Consider the family of subspaces $F=\{f_1,\ldots,f_n\}$, where $f_i:=\sp\{u_i,v_i\}$, for each $i=1,\ldots,n$. 
Let $\x=(x_1,\ldots,x_d)$ and consider $h:=(\sp\{\x\})^\perp$.
In view of Lemma~\ref{fcaph}, we have 
$$\rank A(\x)=d(\{f\cap h\mid f\in F\}).
$$

On the other hand, by Theorem~\ref{Lovthm}, we have $d(\{f\cap h\mid f\in F\})=\rho_1(F)$. Thus  there exists a deterministic strongly-polynomial time algorithm to compute $\rank A(\x)$.
\end{proof}

We note that in the exact same way, our ability to efficiently compute $\rho_k$ for every integer $k$ by Theorem~\ref{main}, and the characterization above, completes the proof of Theorem~\ref{Rk} from the introduction.

\vspace{1.5cm}
\noindent{\bf Acknowledgements} We would like to thank Ze'ev Dvir for many illuminating discussions. We thank Amir Shpilka and Roy Meshulam for useful comments on an earlier version of the paper. We also thank Jan Vondrak for telling us about Dilworth truncation.

\section*{Appendix: Proof of Theorem~\ref{generic} over $\R$}

Here we provide an alternative proof of Theorem~\ref{generic} which works over the field of Real numbers. 
One advantage of working over $\R$ is that we have the notions of a manifold and of the dimension of a manifold available. In the proof below, we use the fact that the set of linear subspaces of $\R^d$ can be viewed as a manifold. Then, to show that a certain set has measure zero, it is sufficient to show that this set has lower dimension. This allows us to obtain a more straightforward proof for the case $\K=\R$.

\begin{proof}[Proof over $\R$:]
We first prove that property (i) in Definition~\ref{def:genp}  is a generic propery.
Fix $A\subset F$ and put $g=\sp(A)$.
For $\x\in\S^{d-1}$ with  $g\subset h(\x)$, we  have $\x\in\mathbb{S}^{d-1}\cap g^{\perp}$. If $d(g)\ge 1$, this means that $\x$ lies in a lower-dimensional sphere, which is a measure-zero subset of $\mathbb{S}^{d-1}$.
Since $F$ is finite (and so the number of different sub-families $A$ is finite), we conclude that for every $\x\in\S^{d-1}$, excluding a finite union of  certain lower-dimensional sub-spheres of $\S^{d-1}$,
$h(\x)$  satisfies property (i) in Definition~\ref{def:genp}.

We now prove that property (ii) in Definition~\ref{def:genp}  is a generic 
property.
Fix some subfamilies $A,B,C\subset F$.
We first handle certain degenerate cases. Note that if
\begin{equation}\label{x=yz}
\sp(A\cap h(\x))=\sp((A\cap h(\x))\cup B)\cap \sp((A\cap h(\x))\cup C),
\end{equation}
for some $\x\in\S^{d-1}$, then $h(\x)$ clearly satisfies property (ii).
Using Lemma~\ref{fcaph}, condition \eqref{x=yz}
defines an algebraic subvariety of $\S^{d-1}$.
In particular, \eqref{x=yz} either holds for every $\x\in \S^{d-1}$ or holds only for $\x$ taken from a subset of $\S^{d-1}$ of measure zero.
In the former case this means that, with respect to the subfamilies $A,B,C$, property (ii) in Definition~\ref{def:genp} holds for $h(\x)$ for every
$\x\in\S^{d-1}$ and there is nothing to prove. Therefore we can assume that we are in the complementary case. Namely, we assume that for almost every $\x\in\S^{d-1}$ we have
\begin{equation}\label{xneqyz}
\sp(A\cap h(\x))\subsetneq \sp((A\cap h(\x))\cup B)\cap \sp((A\cap h(\x))\cup C).
\end{equation}

Our next step is to identify the set of subspaces $g$ of the form 
$g=\sp(A\cap h(\x))$, for some $\x\in\S^{d-1}$, and determine its dimension as a subset of the Grassmannian.

We need the following observation.
Let 
$$
r:=\max_{\x\in \S^{d-1}} d(A\cap h(\x))
$$
We claim that $d(A\cap h(\x))=r$, for almost every $\x\in\S^{d-1}$.
Indeed, by Lemma~\ref{fcaph},
one can write a basis for $\sp(A\cap h(\x))$ with entries that are linear combinations in the coordinates of $\x$.
In particular, $d(A\cap h(\x))$ can be expressed as the rank of a certain symbolic matrix, with entries depending linearly in the coordinates of $\x$.
This implies that $d(A\cap h(\x))=r$ for every $\x\in\S^{d-1}$, excluding some subset of $\S^{d-1}$ of measure zero, which proves our claim. (Here we used the fact that the maximal rank of a given symbolic matrix is the same as the {\it generic} rank of the matrix.)

Let $S_0$ denote the subset of $\x\in \S^{d-1}$ such that either 
$d(A\cap h(\x))<r$ or \eqref{x=yz} holds for $h(\x)$.
As argued above $S_0\subset \S^{d-1}$ has measure zero.
 Let ${\rm Gr}(r,d)$ 
denote the Grassmannian of $r$-dimensional subspaces of $\R^d$, regarded as an affine variety.
 We define a map 
$\phi:\S^{d-1}\setminus S_0\to {\rm Gr}(r,d)$ by  
$$
\x\mapsto \sp(A\cap h(\x)).
$$
We claim that the image of $\phi$ is $r$-dimensional.
Indeed, let $g\in {\rm Im}(\phi)$ and let $\x\in\phi^{-1}(g)$. 
By definition of the domain of $\phi$, we have $\x\not\in S_0$ and thus $d(g)=r$. 
This means $g$ has maximal dimension. Observe that this guarantees that, for every $\x\in g^{\perp}$, we have $\sp(A\cap h(\x))=g$. 
(Indeed, $\x\in g^{\perp}$ certainly implies that $g\subset \sp(A\cap h(\x))$ and since $d(A\cap h(\x))\le r=d(g)$, we have equality.)
That is, $\phi^{-1}(g)=(\S^{d-1} \setminus S_0)\cap g^\perp$ and, in paticular,
$$\dim (\phi^{-1}(g))=d-1-r
$$
(dimension here is as a manifold).
We conclude that 
\begin{equation}\label{dimim}
\dim{\rm  Im}(\phi)=d-1-(d-1-r)=r,
\end{equation}
as claimed.

Next, define
$$
S_1'=\{\x\in\S^{d-1}\mid  
\sp((A\cap h(\x))\cup B)\cap\sp((A\cap h(\x))\cup C)\subset h(\x)
\}.
$$
Our goal is to show that $S_1'$ has measure zero, as a subset of the sphere. For this, it suffices to show that $S_1:=S_1'\setminus S_0$ has measure zero (since $S_0$ is of measure zero). Consider the restriction of $\phi$ to $S_1$.
Let $g\in {\rm Im}(\phi|_{S_1})$ and let $\x\in \phi|_{S_1}^{-1}(g)$.
Set 
$$
g':=\sp((A\cap h(\x))\cup B)\cap\sp((A\cap h(\x))\cup C).
$$
Since $\x\not\in S_0$, we have \eqref{xneqyz} which means
$$
d(g')\ge r+1.$$
Since we assume also that $\x\in S_1'$, we have $\x\in (g')^\perp$.
So 
\begin{equation}\label{fiber'}
\dim (\phi|_{S_1}^{-1}(g))\le d(g')^\perp\le d-1-(r+1)=d-r-2.
\end{equation}
Clearly we also have ${\rm Im}(\phi|_{S_1})\subset {\rm Im}(\phi)$,
and thus, using \eqref{dimim},   
\begin{equation}\label{im'}
\dim ({\rm Im}(\phi|_{S_1}))\le r.
\end{equation}
Combining \eqref{fiber'} and \eqref{im'}, we get that 
$$
\dim S_1=\dim ({\rm Im}(\phi|_{S_1}))+\dim (\phi|_{S_1}^{-1}(g))\le d-2.
$$
This completes the proof of the lemma.\end{proof}

\end{document}